\documentclass[a4paper,11pt]{article}
\usepackage[a4paper,margin=1in]{geometry}
\usepackage[utf8]{inputenc}
\usepackage{amsmath}
\usepackage{amssymb}
\usepackage{amsthm}
\usepackage{booktabs}
\usepackage[small]{caption}
\usepackage{cite}
\usepackage{colortbl}
\usepackage{enumitem}
\usepackage{framed}
\usepackage{graphicx}
\usepackage{microtype}
\usepackage{xcolor}
\usepackage{hyperref}
\usepackage{thm-restate}
\usepackage{cite}

\theoremstyle{plain}
\newtheorem{theorem}{Theorem}[section]
\newtheorem{lemma}[theorem]{Lemma}
\newtheorem{corollary}[theorem]{Corollary}

\theoremstyle{definition}
\newtheorem{definition}[theorem]{Definition}

\theoremstyle{remark}

\newtheorem{openquestion}[theorem]{Open question}

\DeclareMathOperator{\polylog}{polylog}

\DeclareMathOperator{\Start}{start}
\DeclareMathOperator{\End}{end}
\DeclareMathOperator{\Edge}{edge}
\DeclareMathOperator{\Node}{node}
\DeclareMathOperator{\loutput}{Output}
\DeclareMathOperator{\happy}{happy}

\newcommand{\lup}{\mathsf {Up}}
\newcommand{\lleft}{\mathsf {Left}}
\newcommand{\lright}{\mathsf {Right}}
\newcommand{\ldown}{\mathsf {Down}}
\newcommand{\lzero}{\mathsf {Zero}}
\newcommand{\lflag}{\mathsf {Flag}}

\newcommand{\lgrid}{\mathsf {grid}}
\newcommand{\llcl}{\mathsf {LCL}}

\newcommand{\fakeparagraph}[2]{\par\noindent\textbf{#1}\hspace{1em}#2}

\hypersetup{
    colorlinks=true,
    linkcolor=black,
    citecolor=black,
    filecolor=black,
    urlcolor=[rgb]{0,0.1,0.5},
    pdftitle={Local Mending},
    pdfauthor={Alkida Balliu, Juho Hirvonen, Darya Melnyk, Dennis Olivetti, Joel Rybicki, Jukka Suomela}
}

\newcommand{\myemail}[1]{\,$\cdot$\, {\small #1}}
\newcommand{\myaff}[1]{\,$\cdot$\, {\small #1}\par\smallskip}

\newenvironment{myabstract}
{\list{}{\listparindent 1.5em%
        \itemindent    \listparindent
        \leftmargin    0cm
        \rightmargin   0cm
        \parsep        0pt}%
    \item\relax}
{\endlist}

\newenvironment{mycover}
{\list{}{\listparindent 0pt
        \itemindent    \listparindent
        \leftmargin    0cm
        \rightmargin   0cm
        \parsep        0pt}%
    \raggedright
    \item\relax}
{\endlist}

\begin{document}

\begin{mycover}
{\huge\bfseries\boldmath Local Mending \par}
\bigskip
\bigskip
\bigskip

\textbf{Alkida Balliu}
\myemail{alkida.balliu@cs.uni-freiburg.de}
\myaff{University of Freiburg}

\textbf{Juho Hirvonen}
\myemail{juho.hirvonen@aalto.fi}
\myaff{Aalto University}

\textbf{Darya Melnyk}
\myemail{darya.melnyk@aalto.fi}
\myaff{Aalto University}

\textbf{Dennis Olivetti}
\myemail{dennis.olivetti@cs.uni-freiburg.de}
\myaff{University of Freiburg}

\textbf{Joel Rybicki}
\myemail{joel.rybicki@ist.ac.at}
\myaff{IST Austria}

\textbf{Jukka Suomela}
\myemail{jukka.suomela@aalto.fi}
\myaff{Aalto University}

\end{mycover}
\medskip

\begin{myabstract}
\fakeparagraph{Abstract.}
In this work we introduce the graph-theoretic notion of \emph{mendability}: for each locally checkable graph problem we can define its \emph{mending radius}, which captures the idea of how far one needs to modify a partial solution in order to ``patch a hole.''

We explore how mendability is connected to the existence of efficient algorithms, especially in distributed, parallel, and fault-tolerant settings. It is easy to see that $O(1)$-mendable problems are also solvable in $O(\log^* n)$ rounds in the LOCAL model of distributed computing. One of the surprises is that in paths and cycles, a converse also holds in the following sense: if a problem $\Pi$ can be solved in $O(\log^* n)$, there is always a restriction $\Pi' \subseteq \Pi$ that is still efficiently solvable but that is also $O(1)$-mendable.

We also explore the structure of the landscape of mendability. For example, we show that in trees, the mending radius of any locally checkable problem is $O(1)$, $\Theta(\log n)$, or $\Theta(n)$, while in general graphs the structure is much more diverse.
\end{myabstract}
%\medskip

\section{Introduction}\label{sec:intro}

Naor and Stockmeyer~\cite{Naor1995} initiated the study of the following question: given a problem that is locally \emph{checkable}, when is it also locally \emph{solvable}? In this paper, we explore the complementary question: given a problem that is locally checkable, when is it also locally \emph{mendable}?

\subsection{Warm-up: greedily completable problems}

There are many graph problems in which \emph{partial solutions can be completed greedily, in an arbitrary order}. In particular, several classic problems that have been extensively studied in the field of distributed computing fall into this category: the canonical example is vertex coloring with $\Delta+1$ colors in a graph with maximum degree $\Delta$. Any partial coloring can be always completed; the neighbors of a node can use at most $\Delta$ distinct colors, so a free color always exists for any uncolored node. This simple observation has far-reaching consequences:
\begin{itemize}
	
	\item Any such problem \emph{can be solved efficiently} not only in the centralized sequential setting, but also in distributed and parallel settings. For example, for maximum degree $\Delta = O(1)$ any such problem can be solved in $O(\log^* n)$ communication rounds in the usual LOCAL model \cite{Linial1992,Peleg2000} of distributed computing.
	\item Any such problem often admits simple \emph{fault-tolerant and dynamic algorithms}. For example, one can simply clear the labels in the immediate neighborhood of any point of change and then greedily complete the solution.
\end{itemize}
Classic symmetry-breaking problems such as maximal matching and maximal independent set also fall in this class of problems. However, there are problems that admit efficient distributed solutions even though they are not greedily completable.

In this work, we will introduce the notion of \emph{local mendability} that captures a much broader family of problems, and that has the same attractive features as greedily completable problems: it implies efficient centralized, distributed, and parallel solutions, as well as fault-tolerant and dynamic algorithms.

\subsection{Informal example: mending partial colorings in grids}\label{ssec:grid4col}

Let $G$ be a large two-dimensional grid graph; this is a graph with maximum degree $4$. As discussed above, $5$-colorings in such a graph can be found greedily; any partial solution can be completed. However, $4$-coloring is much more challenging. Consider, for example, the partial $4$-coloring in Figure~\ref{fig:4col}a: the unlabeled node in the middle does not have any free color left; this is not greedily completable. Also, the four neighbors of the node do not have any other choices.

However, one can make the empty region a bit larger and create a $2\times2$ hole in the grid, as shown in Figure~\ref{fig:4col}b. This way, we will always create a partial coloring that is completable---this is a simple corollary of more general results by Chechik and Mukhtar \cite{chechik2019optimal}.

To see this, notice that each node in the $2\times2$ part has got at least two possible choices left; hence we arrive at the task of list coloring of a $4$-cycle with lists of size at least $2$ (Figure~\ref{fig:4col}c). It is known that any even cycle is $2$-choosable~\cite{erdos1980choosability}, i.e., list-colorable with lists of size $2$. We can therefore find a feasible solution, e.g., the one shown in Figure~\ref{fig:4col}d.

Note that here to complete the partial coloring, we had to change the label of a node at distance $2$ from the hole, but we never needed to modify the given partial labeling any further than that. Therefore we say that \textbf{\boldmath $4$-coloring in grids is $2$-mendable}; informally, \emph{we can mend any hole in the solution with a ``patch'' of radius at most $2$}. We can contrast this with the $5$-coloring problem, which is $0$-mendable (i.e., greedily completable without touching any of the previously assigned labels), and the $3$-coloring problem, which turns out to be not $T$-mendable for any constant $T$ (i.e., there are partial $3$-colorings that no constant-radius patch will mend).

\medskip
\begin{framed}%\small
\fakeparagraph{Something swept under the carpet for now.} The above example is informal, and when we later formalize the notion of mendability so that it is well-defined for \emph{any} locally checkable problem, not just the particularly convenient problem of graph coloring, it will also change the precise mending radius of graph coloring by an additive $+1$. Both the above informal view and the formalism are fine for understanding our results; we are, after all, primarily concerned about the asymptotics. We return to this aspect later in Section~\ref{sec:mendability}.
\end{framed}

\begin{figure}[t]
        \centering
        \includegraphics[page=1,scale=0.9]{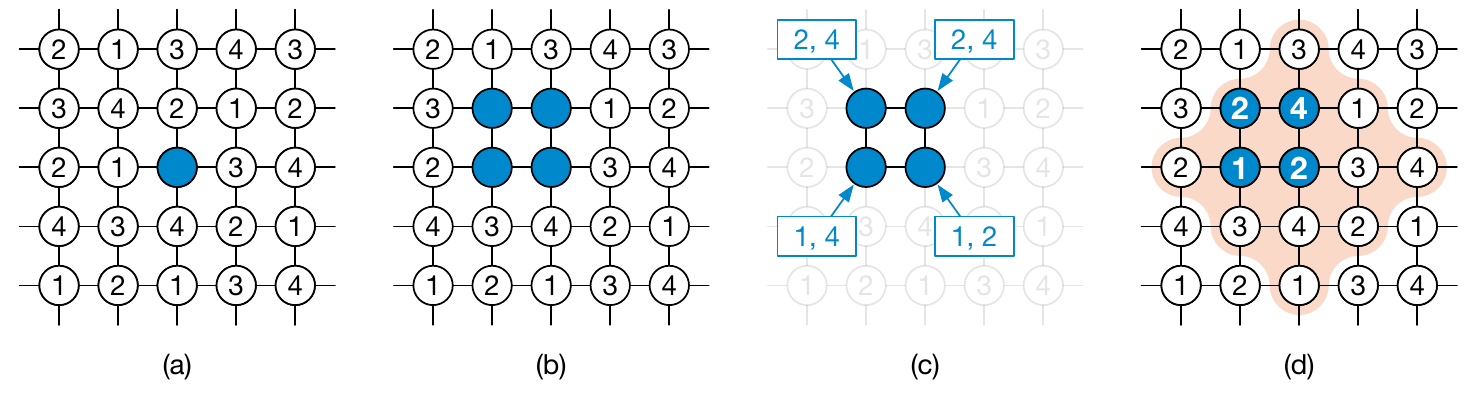}
        \caption{Mending of $4$-colorings in a two-dimensional grid: (a)~A partial solution given as input, with one hole (blue). (b)~The hole enlarged. (c)~This leads to a list coloring instance with lists of size $2$ in a $4$-cycle. (d)~As $4$-cycles are $2$-choosable, we can always complete the coloring in the hole. The orange shading indicates the radius-$2$ neighborhood of the hole, and by following this procedure, it is sufficient to modify the partial input in this region to patch the hole; therefore we say that $4$-coloring in grids is $2$-mendable.}\label{fig:4col}
\end{figure}

\subsection{Consequences of local mendability}\label{ssec:consequences}

Problems that are locally mendable have several attractive properties. In particular, as the notion of mendability does not depend on any specific computational model, it naturally lends itself to algorithm design paradigms in centralized, parallel, and distributed systems. We will now discuss a few motivating examples.

Let us consider any graph problem $\Pi$ that is $T$-mendable for a constant $T$; we formally define this notion in Section~\ref{sec:mendability}, but for now, the informal idea that holes can be mended with radius-$T$ patches will suffice. To simplify the discussion, we will work in bounded-degree graphs, i.e., $\Delta = O(1)$, and assume that we have a discrete graph problem in which nodes are labeled with labels from some finite set of labels; $4$-coloring grids is a simple example of a problem that satisfies all these properties, with $T = 2$.

We can now make use of \emph{mendability as an algorithm design paradigm}. First, we obtain a very simple centralized, sequential linear-time algorithm for solving $\Pi$: process the nodes in an arbitrary order and whenever we encounter a node that has not been labeled yet, mend the hole with a constant-size patch. Mendability guarantees that such a patch always exists: as there are only constantly many candidates, one can find a suitable patch e.g.\ by checking each candidate.

A key observation is that we can \emph{mend two holes simultaneously in parallel} as long as the distance between them is at least $2T+1$; this ensures that the patches do not overlap. This immediately suggests efficient parallel and distributed algorithms: find a set of holes that are sufficiently far from each other, and patch all of them simultaneously in parallel. Concretely, one can first find a distance-${(2T+1)}$ coloring of the graph, and then proceed by color classes, mending all holes in one color class simultaneously in parallel. For example, in the usual LOCAL model~\cite{Linial1992,Peleg2000} of distributed computing, this leads to an $O(\log^* n)$-time algorithm for solving~$\Pi$.

Local mendability also directly leads to very efficient \emph{dynamic graph algorithms} that can maintain a feasible solution in constant time per update (at least for so-called LCL problems, see Section~\ref{sec:definitions}): simply clear the solution around the point of change and patch the hole locally.
This approach holds both in the classic setting of centralized dynamic graph algorithms~\cite{demetrescu2010dynamic,Henzinger18}, but also in so-called distributed input-dynamic settings~\cite{foerster2021input}, where the communication topology remains static but the local inputs may dynamically change.

Furthermore, local mendability can be used to design algorithms that can \emph{recover from failures}. For example, it can be used as a tool for designing \emph{self-stabilizing algorithms}~\cite{Dolev2000,Dijkstra1974self-stabilizing}: in brief, nodes can repeatedly check if the solution is locally correct, switch to an empty label if there are local errors, and then patch any holes in a parallel manner as sketched above.

\subsection{Case study: Mending as an (automatic) algorithm design tool}

So far, we have seen that mendability can be used as a tool to design efficient algorithms; but is this a powerful tool that makes algorithm design \emph{easier}? The problem of $4$-coloring grids provides a rather convincing example: directly designing an efficient $4$-coloring algorithm e.g.\ for the LOCAL model of distributed computing is challenging; none of the prior algorithms~\cite{Brandt2017,panconesi95delta,holroyd2017finitary} are easy to discover or explain, but the above algorithm that makes use of the concept of mending is short and~simple.

In Appendix~\ref{app:134orientation}, we give another example: \emph{$\{1,3,4\}$-orientation}. In this problem, the task is to orient the edges in a two-dimensional grid such that each (internal) node has indegree $1$, $3$, or $4$. This is a problem that served as an example of computational algorithm synthesis in prior work~\cite{Brandt2017}. It is a nontrivial problem for algorithm designers. In fact, it is the most constrained orientation problem that is still $O(\log^* n)$-round solvable---for any $X \subsetneq \{1,3,4\}$ the analogous $X$-orientation problem is \emph{not} solvable in $O(\log^* n)$ rounds~\cite{Brandt2017}.

The prior $O(\log^* n)$-round algorithm for $\{1,3,4\}$-orientation was discovered with computers~\cite{Brandt2017}; this algorithm is in essence a large lookup table, and as such not particularly insightful from the perspective of human beings. It moreover made additional assumptions on the input (specifically, oriented grids).

With the help of mending we can design a much simpler algorithm for the same problem. We show that $\{1,3,4\}$-orientations are locally mendable, and we can patch them as follows:
\begin{itemize}[noitemsep]
	\item Let $G$ be a two-dimensional grid, and let $H$ be the subgraph induced by $3 \times 3$ nodes in $G$.
	\item Regardless of how the edges of $G$ outside $H$ are oriented, we can always orient the edges of $H$ such that all nodes of $H$ have indegree $1$, $3$, or $4$.
\end{itemize}
Note that we can mend any partial orientation by simply patching the orientations in a $3 \times 3$ square around each hole. This way, we arrive at a very simple $O(\log^* n)$-round algorithm for the orientation problem. To show that this works, it is sufficient to show the \emph{existence} of a good orientation of $H$. And since there are only $12$ edges in $H$, it is quick to \emph{find} a suitable orientation for any given scenario.

We give the proof of the mendability of the $\{1,3,4\}$-orientation problem in Appendix~\ref{app:134orientation}. While we give a human-readable proof, we also point out that one can quickly verify that $\{1,3,4\}$-orientations are mendable using a simple computer program that enumerates all possible scenarios (there are only 1296 cases to check). This demonstrates the power of mending also as an \emph{algorithm synthesis} tool: for some problems we can determine automatically with computers that they are locally mendable, and hence we immediately also arrive at simple and efficient algorithms for solving them.

\section{Contributions and key ideas}\label{sec:contrib}

\paragraph{Concept.}
Our first main contribution is the new notions of mendability and the mending radius.
While ideas similar to mendability have been used in numerous papers before this work (see Section~\ref{sec:related}), we are not aware of prior work that has given a definition of mendability that is as broadly applicable as ours:
\begin{itemize}[noitemsep]
	\item Our definition is \emph{purely graph-theoretic}, without any references to any model of computation.
	\item Our definition is applicable to the study of \emph{any locally verifiable graph problem}; this includes all LCL problems (locally checkable labelings), as defined by Naor and Stockmeyer~\cite{Naor1995}.
\end{itemize}
We emphasize that with our definition it makes sense to ask questions like, for example, what is the mendability of the maximal total dominating set problem or the locally optimal cut problem. Moreover, we aim at giving a simple definition that is as robust and universal as possible. As it does not refer to any model of computation, it is independent of the numerous modeling details that one commonly encounters in distributed computing (deterministic or randomized; private or public randomness; whether the nodes know the size of the network; if there are unique identifiers and what their range is; whether the message sizes are limited; synchronous or asynchronous; fault-free or fault-tolerant). The formal definition of mendability is presented and discussed in Section~\ref{sec:mendability}.

\paragraph{From solvability to mendability.}
As our second main contribution, we explore whether efficient mendability is \emph{necessary} for efficient solvability.
As we have already discussed in Section~\ref{ssec:consequences}, it is well-known and easy to see that \emph{local mendability implies efficient solvability}, not only in the centralized setting but also in parallel and distributed settings. One concrete example is the following result (here LCLs are a broad family of locally verifiable graph problems and LOCAL is the standard model of distributed computing; see Section~\ref{sec:definitions} for the details): Let us consider an LCL problem $\Pi$. Then, if $\Pi$ is $O(1)$-mendable, it is also solvable in $O(\log^* n)$ communication rounds in the LOCAL model. The proof of this result is presented in Section~\ref{sec:mendability_to_solvability}.

However, we are primarily interested in exploring the converse: given e.g.\ an LCL problem $\Pi$ that is solvable in $O(\log^* n)$ rounds in the LOCAL model, what can we say about the mendability of $\Pi$? Is local mendability not only sufficient but also necessary for efficient solvability?

It turns out that the answer to this is not straightforward. We first consider the simplest possible case of \emph{unlabeled paths and cycles}. This is a setting in which local solvability is nowadays very well understood \cite{chang21automata-theoretic,Brandt2017,Naor1995}. In the following, we briefly outline the connections between local solvability and mendability that we will show in this work:
\begin{enumerate}
	\item There are LCL problems that are $O(\log^* n)$-round solvable but $\Omega(n)$-mendable in paths and cycles.
	\item However, for every LCL problem $\Pi$ that is $O(\log^* n)$-round solvable in paths and cycles, there exists a \emph{restriction} $\Pi' \subseteq \Pi$ such that $\Pi'$ is $O(\log^* n)$-round solvable and also $O(1)$-mendable. Furthermore, given the description of $\Pi$, there is an efficient algorithm that constructs $\Pi'$.
\end{enumerate}
The second point states that we can always turn any locally solvable LCL problem into a locally mendable one, without making it any harder to solve. In this sense, in paths and cycles, local mendability is \emph{equivalent} to efficient solvability. These results are discussed in Section~\ref{sec:solvability_to_mendability}.

Let us next consider a more general case. The first natural step beyond paths and cycles is rooted trees; for problems similar to graph coloring (more precisely, for so-called edge-checkable problems) their local solvability in rooted trees is asymptotically equal to their local solvability in directed paths \cite{chang21automata-theoretic}, and hence one might expect to also get a very similar picture for the relation between local mendability and local solvability. However, the situation is much more~diverse:
\begin{enumerate}
	\item We start with bad news: the above idea of restrictions does not hold in rooted trees. As a concrete example, let $\Pi$ be the problem of $3$-coloring binary trees. This problem is $O(\log^* n)$-round solvable but not $O(1)$-mendable. We show that, for many natural encodings of $\Pi$, no restriction $\Pi' \subseteq \Pi$ is simultaneously $O(\log^* n)$-round solvable and~$O(1)$-mendable. 
	\item However, one can always augment the problem with auxiliary labels to construct a new problem $\Pi^*$ with the following properties: if $\Pi$ is $O(\log^* n)$-round solvable then $\Pi^*$ is also $O(\log^* n)$-round solvable; furthermore, any solution of $\Pi^*$ can be projected to a feasible solution of $\Pi$, and $\Pi^*$ is $O(1)$-mendable. This works not only in rooted trees but also in general graphs.
	\item In general, the mechanical construction of $\Pi^*$ is rather heavyweight and leads to unnatural problems. However, this idea can be used to derive natural problems that are both easy to mend and easy to solve. We use again $3$-coloring in binary trees as an example: we show how to construct a natural problem that is equivalent to $3$-coloring, that is $O(\log^* n)$-round solvable, and that is also $O(1)$-mendable.
\end{enumerate}
Mendability on rooted trees is discussed in Section~\ref{ssec:mendability_rooted_trees}. The problem of $3$-coloring on binary trees highlights that it is often possible to modify our classical problems such that they remain easily solvable and simultaneously become efficiently mendable. This way, one can turn existing efficient distributed algorithms into fault-tolerant algorithms with many desirable properties (for example, even if the initial configuration is adversarial, recomputation is not needed in regions that are locally feasible).

The key open question in this line of research is the following:
\begin{openquestion}
	Can we develop \emph{efficient, general-purpose techniques} that can turn any locally solvable problem into a concise and natural locally mendable problem?
\end{openquestion}

\paragraph{Landscape of mendability.}
Our third main contribution is that we initiate the study of the \emph{landscape} of mendability, in the same spirit as what has been done recently for the local solvability of LCL problems \cite{Balliu2018stoc,Balliu2018disc,Chang2019,Ghaffari2018,balliu20lcl-randomness,fischer17sublogarithmic,Rozhon2019,Brandt2016,chang16exponential,ghaffari17distributed,Naor1995,balliu20binary-labeling,chang21automata-theoretic,balliu19weak-col,chang2020trees,Brandt2017,balliu19lcl-decidability,suomela2020landscape}. The question that we ask here is simple: what are possible functions $T$ such that there exists an LCL problem $\Pi$ that is $T(n)$-mendable in graphs with $n$ nodes?

In particular, if we can \emph{exclude} the possibility of a broad range of possible functions $T$, this will greatly help us with the analysis of mendability for any given problem. In this work, we show the following results:
\begin{enumerate}
	\item In cycles and paths, there are only two possible classes: $O(1)$-mendable problems and $\Theta(n)$-mendable problems.
	\item In trees, there are exactly three classes: $O(1)$-mendable, $\Theta(\log n)$-mendable, and $\Theta(n)$-mendable problems.
	\item In general bounded-degree graphs, there are additional classes; we show that e.g.\ $\Theta(\sqrt{n})$-mendable problems exist.
\end{enumerate}
These results are presented in Section~\ref{sec:landscape}. The key open question in this line of research is to complete the characterization of mendability in general graphs:
\begin{openquestion}
	Are there LCL problems that are $\Theta(n^\alpha)$-mendable for all rational numbers $0 < \alpha < 1$? For all algebraic numbers $0 < \alpha < 1$? 
\end{openquestion}
\begin{openquestion}
	Are there LCL problems that are $T(n)$-mendable for some $T(n)$ that is between $\log^{\omega(1)} n$ and $n^{o(1)}$?
\end{openquestion}

\paragraph{An application.}

We highlight one example of nontrivial corollaries of our work: in trees, any $o(n)$-mendable problem can be solved in $O(\log n)$ rounds in the LOCAL model. To see this, we can put together the following results:
\begin{enumerate}
	\item Our work on the landscape of mendability shows that $o(n)$-mendability implies $O(\log n)$-mendability in trees.
	\item As we will see, $O(\log n)$-mendable problems can be solved with the help of network decomposition algorithms in $\polylog n$ rounds.
	\item Finally, by applying the known results on the landscape of distributed computational complexity in trees \cite{Chang2019}, one can speed up algorithms that solve LCL problems from $\polylog n$ rounds to $O(\log n)$ rounds.
\end{enumerate}
The formal proof of this result is given in Section~\ref{sec:corollaries}.

\section{Related work}\label{sec:related}

The underlying idea of local mending is not new. Indeed, similar ideas have appeared in the literature over at least three decades, often using terms such as \emph{fixing}, \emph{correcting}, or, similar to our work, \emph{mending}. However, none of the prior work that we are aware of captures the ideas of mendability and mending radius that we introduce and discuss in this work.

Despite some dissimilarities, our work is heavily inspired by the extensive prior work. The graph coloring algorithm by Chechik and Mukhtar~\cite{chechik2019optimal} serves as a convincing demonstration of the power of the mendability as an algorithm design technique. Discussion on the maximal independent set problem in Kutten and Peleg~\cite{Kutten2000FaultLocality} and König and Wattenhofer~\cite{kw13localfix} highlights the challenges of choosing appropriate definitions of ``partial solutions'' for locally checkable problems.

\paragraph{Mending as an ad-hoc tool.}
The idea of mendability has often been used as a tool for designing algorithms for various graph problems of interest.
However, so far the idea has not been introduced as a general concept that one can define for any locally checkable problem, but rather as specific algorithmic tricks that can be used to solve the specific problems at hand.

For example, the key observation of Panconesi and Srinivasan~\cite{panconesi95delta} can be phrased such that $\Delta$-coloring is $O(\log n)$-mendable. In their work, they show how this result leads to an efficient distributed algorithm for $\Delta$-coloring. Aboulkeret al.~\cite{Aboulker2019} and Chierichetti and Vattani~\cite{Chierichetti2010local} study list coloring from a similar perspective. Barenboim~\cite{Barenboim2016Coloring} uses the idea of local correction in the context of designing fast distributed algorithms for the $(\Delta+1)$-coloring problem.

Harris et al.~\cite{harris2020locality} show that an edge-orientation with maximum out-degree $a^*(1+\varepsilon)$ can be solved using ``local patching'' or, in other words, mending; $a^*$ here denotes the pseudo-arboricity of the considered graph. The authors first show that this problem is $O(\log n / \varepsilon)$-mendable. They then use network decomposition of the graph and the mending procedure to derive an $O(\log^3 n / \varepsilon)$ round algorithm for $a^*(1+\varepsilon)$-out-degree orientation.

Chechik and Mukhtar~\cite{chechik2019optimal} present the idea of ``removable cycles'', and show how it leads to an efficient algorithm for $4$-coloring triangle-free planar graphs (and $6$-coloring planar graphs). When one applies their idea to $2$-dimensional grids, one arrives at the observation that $4$-coloring in grids is $2$-mendable, as we saw in the warm-up example of Section~\ref{ssec:grid4col}.

Chang et al.~\cite{Chang2018edgeColoring} study similar ideas from the complementary perspective: they show that edge colorings with few colors are \emph{not} easily mendable, and hence mending cannot be used (at least not directly) as a technique for designing efficient algorithms for these problems.

\paragraph{Mending with advice.}
The work by Kutten and Peleg~\cite{Kutten1999Mending,Kutten2000FaultLocality} is very close in spirit to our work. They study the \emph{fault-locality} of mending: the idea is that the cost of mending only depends on the number of failures, and is independent of the number of nodes. However, one key difference is that they assume that the solution of a graph problem is augmented with some auxiliary precomputed data structure that can be used to assist in the process of mending. They also show that this is highly beneficial: \emph{any} problem can be mended with the help of auxiliary precomputed data structures. We note that the addition of auxiliary labels is similar in spirit to the brute-force construction discussed in Section \ref{ssec:mendability_rooted_trees} that turns any locally solvable problem into an equivalent locally mendable one.

Censor-Hillel et al.~\cite{censorhillel2020fast} also comes close to our work with their definition of locally-fixable labelings (LFLs). In essence, LFLs are specific LCLs that can be (by construction) mended fast. In our work we seek to understand which LCLs can be mended fast, and also look beyond the case of constant-radius mendability.

Mending with the help of advice is also closely connected to distributed verification with the help of \emph{proofs}. The idea is that a problem is not locally verifiable as such, but it can be made locally verifiable if we augment the solution with a small number of proof bits; see e.g.\ Göös and Suomela~\cite{goos11lcp}, Korman and Kutten~\cite{korman06distributed,korman07distributed}, Korman et al.~\cite{korman10proof,korman10constructing}.

\paragraph{Mending small errors.}
Another key difference between our work and that of e.g.\ Kutten and Peleg~\cite{Kutten1999Mending,Kutten2000FaultLocality} is that we want to be able to mend holes in \emph{any} partial solution, including one in which there is a linear number of holes. In particular, mending as defined in this work is something one can use to efficiently compute a correct solution from scratch.

In the work by Kutten and Peleg~\cite{Kutten1999Mending,Kutten2000FaultLocality}, the cost of fixing can be very high if the number of faults is e.g.\ linear. The work by König and Wattenhofer~\cite{kw13localfix} is similar in spirit: they assume that holes are small and well-separated in space or time.

\paragraph{Making algorithms fault-tolerant.}
There is an extensive line of research of systematically turning existing efficient distributed algorithms into  dynamic or fault-tolerant algorithms; examples include Afek and Dolev~\cite{Afek1997LocalStabilizer}, Awerbuch et al.~\cite{Awerbuch1991Self-stabilization}, Awerbuch and Sipser~\cite{Awerbuch1988dynamic}, Awerbuch and Varghese~\cite{Awerbuch1991DistributedProgramChecking}, Ghosh et al.~\cite{ghosh2007fault}, and Lenzen et al.~\cite{Lenzen2009local}. This line of research is distinct from our work, where the starting point is a property of the problem (mendability); then both efficient distributed algorithms and fault-tolerant algorithms follow.

\section{Defining mendability}\label{sec:mendability}

\paragraph{Starting point: edge-checkable problems.}
In Section~\ref{sec:intro}, we chose vertex coloring as an introductory example. Vertex coloring is an \emph{edge-checkable} problem, i.e., the problem can be defined in terms of a relation that specifies what label combinations at the endpoints of each edge are allowed. For such a problem the notion of \emph{partial solutions} is easy to define.

More precisely, let $G = (V,E)$ be a graph, and let $\Pi$ be the problem of finding a $k$-vertex coloring in $G$; a feasible solution is a labeling $\lambda\colon V \to \{1,2,\dotsc,k\}$ such that for each edge $\{u,v\} \in E$, we have $\lambda(u) \ne \lambda(v)$. We can then easily \emph{relax} $\Pi$ to the problem $\Pi^*$ of finding a \emph{partial} $k$-vertex coloring as follows: a feasible solution is a function $\lambda\colon V \to \{\bot,1,2,\dotsc,k\}$ such that for each edge $\{u,v\} \in E$, we have $\lambda(u) = \bot$ or $\lambda(v) = \bot$ or $\lambda(u) \ne \lambda(v)$. Here we used $\bot$ to denote a missing label. The same idea could be generalized to \emph{any} problem $\Pi$ that we can define using some edge constraint $C$:
\begin{itemize}[noitemsep]
    \item $\Pi$: function $\lambda\colon V \to \Gamma$ such that each edge $\{u,v\} \in E$ satisfies $\bigl(\lambda(u), \lambda(v)\bigr) \in C$.
    \item $\Pi^*$: function $\lambda\colon V \to \{\bot\} \cup \Gamma$ such that each edge $\{u,v\} \in E$ with $\bot \notin \{\lambda(u), \lambda(v)\}$ satisfies $\bigl(\lambda(u), \lambda(v)\bigr) \in C$.
\end{itemize}
Note that our original definition $\Pi$ was edge-checkable, and we arrived at a definition $\Pi^*$ of partial solutions for $\Pi$ such that $\Pi^*$ is still edge-checkable. This is important---the problem definition itself remained as local as the original problem.

However, most of the problems that we encounter e.g.\ in the theory of distributed computing are \emph{not} edge-checkable; we will next discuss how to deal with any locally verifiable problem.

\paragraph{Locally verifiable problems.}
In general, a \emph{locally verifiable problem} $\Pi$ is defined in terms of a set of input labels $\Sigma$, a set of output labels $\Gamma$, and a \emph{local verifier} $\psi$. In $\Pi$ we are given a graph $G = (V,E,\sigma)$ with a vertex set $V$, an edge set $E$, and some input labeling $\sigma\colon V \to \Sigma$, and the task is to find an output labeling or \emph{solution} $\lambda\colon V \to \Gamma$ that makes the verifier $\psi$ \emph{happy} at each node $v \in V$. In general a verifier $\psi$ is a function that maps a $(G,\lambda,v)$ to ``happy'' or ``unhappy'', and we say that $\psi$ \emph{accepts} $(G,\lambda)$ if $\psi(G,\lambda,v) = \happy$ for all $v \in V$.

Finally, verifier $\psi$ is a \emph{local} verifier with verification radius $r \in \mathbb{N}$, if $\psi(G,\lambda,v)$ only depends on the input and output within radius-$v$ neighborhood of $v$. That is, $\psi(G,\lambda,v) = \psi(G',\lambda',v')$ if the radius-$r$ neighborhood of $v$ in $G$ (together with the input and output labels) is isomorphic to the radius-$r$ neighborhood of $v'$ in $G'$. Note that we can also generalize the definitions in a straightforward manner from node labelings to edge labelings.

Now edge-checkable problems are clearly locally verifiable problems, with verification radius $r = 1$. But there are also numerous problems that are locally verifiable yet not edge-checkable; examples include maximal independent sets (note that independence is edge-checkable while maximality is not) and more generally ruling sets, minimal dominating sets, weak coloring (at least one neighbor has to have a different label), distance-$k$ coloring (all other nodes within distance $k$ must have different labels), and many other constraint satisfaction problems.

\paragraph{Partial solutions of locally verifiable problems.}

To capture the mendability of a locally verifiable problem, we first need to have an appropriate notion of partial solutions. Ideally, we would like to be able to handle \emph{any} locally verifiable problem $\Pi$ and define a relaxation $\Pi^*$ of $\Pi$ with all the desirable properties as what we had in the graph coloring example:
\begin{enumerate}[noitemsep,label={(P\arabic*)}]
    \item\label{prop1} Problem $\Pi^*$ captures the intuitive idea of partial solutions for $\Pi$, and it serves the purpose of forming the foundation for the notion of mendability and mending radius of $\Pi$.
    \item\label{prop2} An empty solution (all nodes labeled with $\bot$) is a feasible solution for $\Pi^*$.
    \item\label{prop3} Problem $\Pi^*$ is a relaxation of $\Pi$: any feasible solution of $\Pi$ is a feasible solution for $\Pi^*$.
    \item\label{prop4} A feasible solution for $\Pi^*$ without any empty labels is also a feasible solution for $\Pi$.
    \item\label{prop5} The definition of $\Pi^*$ is exactly as local as the definition of $\Pi$: if $\Pi$ is defined in terms of labelings in the radius-$r$ neighborhoods, so is $\Pi^*$.
\end{enumerate}

It turns out that there is a definition with all of these properties, and it is surprisingly simple to state. Let $\Pi$ be a locally verifiable problem with the set $\Gamma$ of output labels and local verifier $\psi$ with verification radius $r$. By definition, $\psi(G,\lambda,v)$ only depends on the radius-$r$ neighborhood of $v$.

We can define a new verifier $\psi^*$ that extends the domain of $\psi$ in a natural manner so that $\psi^*(G,\lambda,v)$ is well-defined also for partial labelings $\lambda \colon V \to \Gamma^*$, where $\Gamma^* = \{\bot\} \cup \Gamma$, as follows:
\begin{framed}
\noindent If there is a node $u$ within distance $r$ from $v$ with $\lambda(u) = \bot$, let $\psi^*(G,\lambda,v) = \happy$.
Otherwise let $\lambda' \colon V \to \Gamma$ be any function that agrees with $\lambda$ in the radius-$r$ neighborhood of $v$, and let $\psi^*(G,\lambda,v) = \psi(G,\lambda',v)$.
\end{framed}
\noindent Note that such a $\lambda'$ always exists, and $\psi^*$ is independent of the choice of $\lambda'$, so $\psi^*$ is well-defined. Furthermore, if $\lambda\colon V \to \Gamma$ is a complete labeling, then $\psi^*$ and $\psi$ agree everywhere, and an empty labeling makes $\psi^*$ happy everywhere. Finally, $\psi^*(G,\lambda,v)$ only depends on the radius-$r$ neighborhood of $v$. If we now define problem $\Pi^*$ using the local verifier $\psi^*$, it clearly satisfies properties \ref{prop2}--\ref{prop5}. Let us now see how to use it to formalize the notion of mendability, and this way also establish \ref{prop1}.

\paragraph{Mendability of local verifiers.}

We first define mendability with respect to a specific verifier $\psi$; as before, $\psi^*$ is the relaxation that also accepts partial labelings. Our definition is minimalistic: we only require that we can make \emph{one unit of progress} by turning any given empty label into a non-empty label---this will be both sufficient and convenient. The key definition is this:
\begin{framed}
\noindent
Let $\lambda \colon V \to \Gamma^*$ be a partial labeling of $G$ such that $\psi^*$ accepts $\lambda$. We say that $\mu \colon V \to \Gamma^*$ is a \emph{$t$-mend of $\lambda$ at node $v$} if:
\medskip
\begin{enumerate}[nolistsep]
    \item $\psi^*$ accepts $\mu$,
    \item $\mu(v) \neq \bot$,
    \item $\mu(u) = \bot$ implies $\lambda(u) = \bot$,
    \item $\mu(u) \neq \lambda(u)$ implies that $u$ is within distance $t$ of $v$.
\end{enumerate}
\end{framed}
That is, in $\mu$ we have applied a radius-$t$ patch around node $v$. If there are some other empties around $v$, they can be left empty (note that this will never make mending harder, as empty nodes only help to make $\psi^*$ happy, and this will not make mending substantially easier, either, as we will need to be able to eventually also patch all the other holes).

Fix a graph family $\mathcal{G}$ of interest. We now define the mending radius of $\psi$ for the entire graph family:
\begin{definition}
Let $T \colon \mathbb{N} \to \mathbb{N}$ be a function. We say that a verifier $\psi$ is $T$-mendable if for all graphs $G \in \mathcal{G}$ and all partial labelings $\lambda$ accepted by $\psi^*$, there exists a $T(|V|)$-mend of $\lambda$ at $v$ for any $v \in V$.
\end{definition}
Note that in general, the mending radius may depend on the number of nodes in the graph; it is meaningful to say that, for example, $\psi$ is $\Theta(\log n)$-mendable, i.e., the mending radius of $\psi$ is $\Theta(\log n)$. We will use throughout this work $n = |V|$ to denote the number of nodes in the input graph.

\paragraph{Mendability of locally verifiable problems.}

So far we have defined the mendability of a particular verifier $\psi$. In general, the same graph problem $\Pi$ may be defined in terms of many equivalent verifiers $\psi$ (for example, vertex coloring can be locally verified with a local verifier $\psi$ that checks the coloring in the radius-$7$ neighborhoods, even if this is not necessary). Indeed, there are problems for which it is not easy to define a ``canonical'' verifier, and it is not necessarily the case that the smallest possible verification radius coincides with the smallest possible mending radius. Hence we generalize the idea of mendability from local verifiers to locally verifiable problems in a straightforward manner:

\begin{definition}
A problem $\Pi$ is \emph{$T$-mendable} if for some $r \in \mathbb{N}$ there exists a $T$-mendable radius-$r$ verifier for the problem $\Pi$. The \emph{mending radius} of a problem $\Pi$ is $T$ if $\Pi$ is $T$-mendable but not $T'$-mendable for any $T' \neq T$ such that $T'(n) \le T(n)$. 
\end{definition}

\paragraph{Discussion and examples.}

Now we have formally defined the mendability of any locally verifiable problem, in a way that makes it applicable to any locally checkable problem. The definition is a natural generalization of the informal idea of mendability for vertex coloring discussed in the introduction. However, it is good to stop here and carefully look at what all of this means in the context of some concrete problem.

Let $\Pi$ be the problem of coloring paths and cycles with $3$ colors. This is a locally verifiable problem, and we can define a local verifier $\psi$ with verification radius $r = 1$ so that $\psi(G,\lambda,v) = \happy$ if all neighbors of $v$ have colors different from $v$. If we have a path with e.g.\ $n = 4$ nodes, $\psi$ will accept labelings that encode a proper coloring (e.g., $1,2,1,3$) and reject improper colorings (e.g., $2,1,1,3$).

However, the relaxed verifier $\psi^*$ is more interesting: it will accept not only labelings like $\bot,\bot,\bot,\bot$ and $\bot,2,1,\bot$ but also labelings like $\bot,1,1,\bot$ (the empty labels next to the nodes labeled with $1$s make the verifier happy). It is easy to see that $\psi$ is $1$-mendable. For example, let
\[
    \lambda_0 = \bot, 1, 1, \bot, \quad
    \lambda_1 = 2, 3, 1, \bot, \quad
    \lambda_2 = 2, 3, 1, 2.
\]
Here $\lambda_1$ is a $1$-mend of $\lambda_0$ at the first node, and $\lambda_2$ is a $0$-mend of $\lambda_0$ at the last node. In general, to fill an empty slot we may need to touch other labels within distance $1$, but not further.

At first, it may seem like a bad idea that e.g.\ the mendability radius of $3$-coloring in paths is $1$ and not $0$, and one may want to start to look for alternative definitions that would impose more stringent constraints for partial solutions. But it turns out that this is not necessary as long as one is only interested in mendability in the asymptotic sense---in essence, all reasonable alternative definitions are equivalent modulo additive constants in the mending radius!

To see this, consider a partial solution $\lambda$ that makes a radius-$r$ verifier $\psi^*$ happy in our sense. Consider now an alternative radius-$r$ verifier $\psi^\dagger$ that is less forgiving near the holes (but is still based on some local rule, and accepts a completely empty solution). Now we can simply ``expand each hole'' in $\lambda$ by $r$ steps and arrive at a solution that will also make $\psi^\dagger$ happy. A bit more precisely, a black-box mending rule $M$ that is applicable w.r.t.\ $\psi^\dagger$ is also applicable w.r.t.\ $\psi^*$: to patch a hole, expand it first, and repeatedly apply $M$ within the expanded hole to fill it in. This will only increase the mending radius by an additive $+r$.

In particular, whatever good news related to the applicability of mending in the design of distributed, parallel, fault-tolerant, and dynamic algorithms we had with the alternative definition $\psi^\dagger$, we have got the same good news also for our minimalistic definition. This is a key reason why we argue that our specific definition is universal: it is applicable to any problem in a very straightforward manner, yet it brings (asymptotically) all the same benefits as any other definition that might be more customized to a specific problem family.

Now we have finally formalized the concept of mendability. We will next briefly recall the definitions of LCL problems and the LOCAL model of computing, and we can then move on to proving the results that we already discussed in Section~\ref{sec:contrib}.

\section{Additional definitions}\label{sec:definitions}

\paragraph{LCL problems.}

We have defined mendability for any locally verifiable problems, but a particularly important special case of locally verifiable problems is the family of \emph{locally checkable problems} (LCL problems), as defined by Naor and Stockmeyer~\cite{Naor1995}. We say that a locally verifiable problem $\Pi$ on a graph family $\mathcal{G}$ is an LCL problem if we have that
\begin{enumerate}[noitemsep]
	\item the set $\Sigma$ of input labels and the set $\Gamma$ of output labels are both finite,
	\item $\mathcal{G}$ is a family of bounded-degree graphs, i.e., there is some constant $\Delta$ such that for any node $v$ in any graph $G \in \mathcal{G}$ the degree of $v$ is at most $\Delta$.
\end{enumerate}
Note that an LCL problem always has a \emph{finite description}: we can simply list all possible non-isomorphic labeled radius-$r$ neighborhoods and classify them based on whether they make the local verifier $\psi$ happy.

\paragraph{LOCAL model.}

Mendability is independent of any model of computing. However, the key applications for the concept are in the context of distributed graph algorithms. For concreteness, we use the LOCAL model~\cite{Linial1992,Peleg2000} of distributed computing throughout this work. In this model, the distributed system is represented as a graph $G=(V,E)$, where each node $v \in V$ denotes a processor and every edge $\{u,v\} \in E$ corresponds to a direct communication link between nodes $u$ and $v$. At the start of the computation, every node receives some local input.

The computation proceeds synchronously in discrete communication rounds and each round consists of three steps: (1) all nodes send messages to their neighbors, (2) all nodes receive messages from their neighbors, and (3) all nodes update their local state. In the last round, all nodes declare their local output. An algorithm has running time $T$ if all nodes can declare their local output after $T$ communication rounds. The bandwidth of the communication links is unbounded, i.e., in each round nodes can exchange messages of any size.
We say that a problem is \emph{$T$-solvable} if it can be solved in the LOCAL model in $T(n)$ communication rounds.

\section{From mendability to solvability}\label{sec:mendability_to_solvability}

In this section, we show that, in some cases, a bound on the mending radius implies an upper bound on the time complexity of a problem in the LOCAL model. Hence, the concept of mendability can be helpful in the process of designing algorithms in the distributed setting.

We start by proving a generic result, that relates mendability with network decomposition; we will make use of the following auxiliary lemma.

\begin{lemma}\label{lem:normalizedmending}
	Let $\Pi$ be an LCL problem with mending radius $f(n)$ and checkability radius $r=O(1)$. Then we can create a mending procedure that only depends on the nodes at distance $f(n)+r$ from the node $u$ that needs to be mended, and it does not even need to know $n$.
\end{lemma}
\begin{proof}
	We first show that it is only needed to inspect the $(f(n)+r)$-radius neighborhood of $u$. We start from $u$ and we inspect its $(f(n) + r)$-radius neighborhood, where $r=O(1)$ is the LCL checkability radius. Since we know that there exists a $f(n)$-mend at $u$, and since the output of a node $v$ may only affect the correctness of the outputs of the nodes at distance at most $r$ from $v$, then it is possible to find a correct mending by brute force. We now remove the dependency on $n$ as follows. We start by gathering the neighborhood of $u$ at increasing distances and at each step we check if there is a feasible mending by brute force. This procedure must stop after at most $f(n)+r$ steps, since we know that such a solution exists.
\end{proof}

We now show that \emph{network decompositions} can be used to relate the mending radius of a problem with its distributed time complexity.
A \emph{$(c,d)$-network decomposition} is a partition of the nodes into $c$ color classes such that within each color class, each connected component has diameter at most $d$~\cite{Awerbuch1996}. Also, recall that $G^i$, the $i$-th power of $G = (V,E)$, is the graph $(V,E')$ satisfying that $\{u,v\} \in E'$ if and only if $u$ and $v$ are at distance at most $i$ in $G$.

\begin{theorem}\label{thm:mending_to_nd}
	Let $\Pi$ be an LCL problem with mending radius $k$ and checkability radius~$r$. Then $\Pi$ can be solved in $O(cd(k+r))$ rounds in the LOCAL model if we are given a $(c,d)$-network decomposition of $G^{2k+r}$. 
\end{theorem}
\begin{proof}
	We prove the claim by providing an algorithm for solving $\Pi$.
	We start by temporarily assigning $\bot$ to all nodes. Then, we process the nodes in $c$ phases. In phase $i$, we mend all nodes that are in components of color $i$ in parallel as follows. By \autoref{lem:normalizedmending}, for each node $v$, we do not need to see the whole graph to find a valid mending of radius $k$ for $v$, but only nodes that, in $G$, are at distance at most $k+r$ from $v$. This implies that we can find a valid mending for all nodes of each component by gathering the whole component and the nodes at a distance of at most $k+r$ from it. This mending only needs to modify the solution of nodes inside the component and nodes at distance at most $k$ from it. Since the network decomposition is done on $G^{2k+r}$, then, in $G$, nodes of different components are at distance strictly larger than $2k+r$ from each other. This implies that the mending applied on some component $C_1$ does not modify the temporary solution of nodes at distance of at most $k+r$ from some other component $C_2 \neq C_1$.
	Hence, we obtain the same result that we would have obtained by mending each component of color $i$ sequentially. Since we process all color classes and perform a valid mending, at the end, no node is labeled $\bot$, and hence the temporary labeling is a valid solution for $\Pi$. Each connected component has diameter at most $O(d(k+r))$ in $G$, so each phase requires $O(d(k+r))$ rounds. The total running time is $O(cd(k+r))$.
\end{proof}

As a corollary of this theorem, we show that in order to prove an $O(\log^* n)$ upper bound on the time complexity of a problem, it is enough to prove that a solution can be mended by modifying the labels within a constant distance.
\begin{corollary}
	Let $\Pi$ be an LCL problem with constant mending radius. Then $\Pi$ can be solved in $O(\log^* n)$ rounds in the LOCAL model. 
\end{corollary}
\begin{proof}
	We prove the claim by providing an algorithm running in $O(\log^* n)$. Let $k = O(1)$ be the mending radius of $\Pi$ and $r$ be its checkability radius. We start by computing a distance-$(2k+r)$ coloring using a palette of $c = \Delta^{2k+r} + 1 = O(1)$ colors, that can be done in $O(\log^* n)$ rounds. Note that such a coloring is a $(c,1)$ network decomposition of $G^{2k+r}$, and we can hence apply \autoref{thm:mending_to_nd} to solve $\Pi$ in constant time.
\end{proof}

\section{Making problems mendable}\label{sec:solvability_to_mendability}

In Section~\ref{sec:mendability_to_solvability}, we saw that local mendability implies local solvability. In this section, we consider the converse: does local solvability imply local mendability? First, we show that mending can be much harder than solving by considering an edge-checkable problem on undirected paths.

\begin{theorem}\label{thm:path_logstar_linear}
	There are LCLs that can be solved in $O(\log^*n)$ time that require $\Theta(n)$ distance for mending.
\end{theorem}
\begin{proof}
	Consider the following LCL problem $\Pi$ on undirected paths. Nodes have no input
	and must produce one of the possible outputs in $\Gamma = \{A,B,1,2,3\}$. The LCL constraints are defined by providing a list of valid edge neighborhoods $C_E = \{\{A,B\},\{1,2\},\{1,3\},\{2,3\}\}$.
	In other words, nodes can either $2$-color the path using labels $\{A,B\}$, or $3$-color the path using labels $\{1,2,3\}$, but they cannot mix the labels in the solution.
	
	Solving this problem requires $\Theta(\log^* n)$ time, as it is necessary and sufficient to produce a $3$-coloring. We now prove that this LCL is $\Theta(n)$-mendable.
	Consider a path $P = (p_0,p_1,\ldots)$ of length $n=2k+1$ and the following partial solution on this path:
	\begin{itemize}[noitemsep]
		\item All nodes $p_i$ such that $i < k$ and $i \equiv 0 \pmod 2$ are labeled $A$.
		\item All nodes $p_i$ such that $i < k$ and $i \equiv 1 \pmod 2$ are labeled $B$.
		\item Node $p_{k}$ is labeled $\bot$.
		\item All nodes $p_i$ such that $i > k$ and $i \equiv 0 \pmod 2$ are labeled $B$.
		\item All nodes $p_i$ such that $i > k$ and $i \equiv 1 \pmod 2$ are labeled $A$.
	\end{itemize}
	Note that there are two regions that are labeled with a valid $2$-coloring, these regions are separated by a node labeled $\bot$, and the two $2$-colorings are not compatible, meaning that it is not possible to assign a label to $p_k$ such that the LCL constraints are satisfied on both its incident edges.
	
	We argue that mending this solution requires linear distance. Observe first that mending this solution using labels $\{1,2,3\}$ would require us to undo the labelings of all nodes. This is because the LCL constraints require that no node of the graph should be labeled $\{A,B\}$ in order to use labels $\{1,2,3\}$. Since there are nodes labeled $\{A,B\}$ at distance $k=\Theta(n)$ from $p_k$, changing the labels of these nodes would require linear time. The remaining option is to mend the solution by only using labels $\{A,B\}$. In this case, at least half of the nodes of the path need to be relabeled in order to produce a valid $2$-coloring and satisfy the constraints.
\end{proof}

\subsection{From local solvability to local mendability: the case of cycles}

In Theorem~\ref{thm:path_logstar_linear}, we showed that local solvability does not imply local mendability on undirected paths. The presented counterexample can be modified to also hold for directed paths or cycles by defining the edge constraints more carefully. The main idea of the counterexample was to use two different sets of labels $\{A,B\}$ and $\{1,2,3\}$ that cannot be both part of the solutions. In order to make this problem efficiently mendable, we could try to remove the set $\{A,B\}$ from the set of possible labels and only work with the labels $\{1,2,3\}$. The restricted problem would still be $O(\log^* n)$-solvable and, in fact, such a restriction would be sufficient to make the problem $O(1)$-mendable. 

In this section, we will consider restrictions of a given LCL problem $\Pi$ under which the problem becomes efficiently mendable. We will define such restrictions with respect to the diagram representations of the respective LCL problems using results from~\cite{Brandt2017,chang21automata-theoretic}. In the diagram representation, an LCL problem is described as a graph $D$ that encodes feasible solutions of the LCL problem. The nodes of $D$ represent the edge-constraints of the LCL problem, and the edges correspond to the node-constraints, i.e. the feasible radius-$r$ neighborhoods of the LCL problem. A walk through the graph $D$ defines a feasible labeling of the LCL problem. A formal definition of the diagram representation of LCLs is given in Appendix~\ref{app:intro_automata}.  

\begin{theorem}\label{thm:cycle_mend_hard_to_easy}
	Suppose $\Pi$ is an LCL problem on directed paths or cycles with no input. If $\Pi$ is $O(\log^* n)$-solvable, we can define a new LCL problem $\Pi'$ with the same round complexity, such that a solution for $\Pi'$ is also a solution for $\Pi$, and $\Pi'$ is $O(1)$-mendable. 
\end{theorem}
\begin{proof}
	Let $D$ be the diagram of $\Pi$. Since $\Pi$ is $O(\log^* n)$-solvable, it must contain at least one flexible state, that is, a state, to which we can return in $k, k+1, k+2, k+3,\ldots$ steps for some constant $k$. The constant $k$ is called the flexibility radius of the respective state. Note that the flexible state would be a self-loop, if $\Pi$ is $O(1)$-solvable. Let $f$ be the flexible state contained in the smallest possible strongly connected component and break ties arbitrarily. Let $S$ be the new diagram induced by the smallest strongly connected component containing~$f$. Let $\Pi'$ be the LCL induced by $S$. 
	
	Since the $f$ is contained in both diagrams, $\Pi'$ has the same complexity as $\Pi$. We now prove that $\Pi'$ is $O(1)$-mendable. Let $k = O(1)$ be the flexibility radius of $\Pi'$ and let $r$ be the size of the neighborhoods of the diagram. In order to mend the solution of some node $u$, we can first delete the solution of all nodes at distance at most $k+r$ from $u$, fix any valid labeling for nodes at distance $k+1,\ldots,k+r$, and rewrite a solution by following the path $(p_1,\ldots,p_{2k+1})$ on the diagram, where $p_1$ and $p_{2k+1}$ are labeled with the same neighborhoods of nodes at distance $k$ from $v$.
\end{proof}

Note that similar results can be derived for undirected paths and cycles. These cases require more restrictive diagrams in order to satisfy $O(\log^* n)$ solvability. In Appendix~\ref{app:restrictions:undirected_paths} and~\ref{app:restrictions:undirected_cycles}, we show how such diagrams have to be restricted in order for the problem to be $O(1)$-mendable.

\subsection{From local solvability to local mendability: the case of rooted trees}\label{ssec:mendability_rooted_trees}

In this section, we will consider mendability of rooted trees. In~\cite{chang21automata-theoretic}, it has been shown that the time complexity of so-called \emph{edge-checkable} LCL problems on rooted trees is asymptotically the same as the time complexity on directed paths. We cannot derive a similar statement for mendability. Consider therefore the $3$-coloring problem on rooted binary trees. Its diagram representation consists only of flexible states that form a strongly connected component. This problem is therefore $O(1)$-mendable on directed paths. It is, however, not $o(\log n)$-mendable on rooted binary trees: if the children of a node have different colors, this node can only be colored with the one remaining color. This means that the leaf nodes can define a unique coloring of the whole tree. Assume now that the root of such a tree is uncolored and that the children and a parent of a node have $3$ different colors. In order to mend the color of the node, we would need to undo the coloring of all nodes in one of its subtrees, thus requiring a mending radius of $\Omega(\log n)$. 

Alternatively, we can instead consider the $3$-coloring problem as a set of possible child-parent configurations. Analogously to the previous section, we could restrict the problem to a subset of possible configurations and mend the problem on the restricted instances. However, independent of how we select such a subset, there still exists an assignment of labels to the leaf nodes that would define a unique coloring of the tree. A formal proof of this observation is given in Appendix~\ref{app:restrictions:rooted_trees}.

There is, however, a straightforward approach that makes any locally solvable problem also locally mendable: we can apply the standard way for solving an $O(\log^* n)$-computable problem $\Pi$ by first computing a distance-$k$ coloring. The local identifiers formed by the distance-$k$ coloring can then be used to solve the problem in a constant number of rounds. To make the problem locally mendable, we require $\Pi'$ to output a distance-$k$ coloring and the output from the constant round algorithm for problem $\Pi$. When mending an instance of $\Pi'$, we need to mend the labels from the distance-$k$ coloring and then complete the solution for the original problem. Both steps can be completed in a constant number of rounds. In the following theorem, we present an idea for a restriction to $\Pi'$ for the $3$-coloring problem on binary trees. The presented approach also makes use of additional labels, but in a simpler way than the discussed straightforward approach.

\begin{theorem}\label{thm:Three-coloring_binary_trees}
	Given the $3$-coloring problem $\Pi$ on rooted binary trees without input with any mending complexity, we can define a new LCL problem $\Pi'$ with the same round complexity, such that a solution for $\Pi'$ is also a solution for $\Pi$, and $\Pi'$ is $O(1)$-mendable. 
\end{theorem}
\begin{proof}
	
	In order to be able to find a successful restriction, we use an auxiliary labeling. In particular, we will differentiate between nodes that have children of the same color and those that do not. We call a node \emph{monochromatic}, if its children are colored with the same color, and \emph{mixed} otherwise. We further assume that the leaves of the tree are monochromatic nodes. We define a new problem $\Pi'$ by restricting the original problem $\Pi$ such that the mixed nodes in the partial coloring must form an independent set.
	
	Note that under such a restriction, $\Pi'$ can be still solved in $O(\log^* n)$ rounds. This solution is based on a $4$-coloring algorithm and a subsequent ``shift down'' strategy, where, in one round, all nodes pick the color of their parent node. Thus, the algorithm always computes a solution with only monochromatic nodes.
	
	In order to show that the solution of $\Pi'$ is $O(1)$-mendable, we will consider two possible cases. Assume first that node $u$ that is to be mended is a monochromatic node. In this case, $u$ can always pick the same color as its sibling thus extending the coloring without introducing any mixed nodes. If $u$ is a mixed node, then the children of $u$ must have different colors. If both children are monochromatic nodes, there must exist a common color that both children can pick, and thus turn $u$ into a monochromatic node. If at least one of the children is a mixed node, $u$ cannot also be mixed, as $\Pi'$ would not be satisfied in this case. Let $v$ be this mixed child of $u$. Observe that the children of $v$ must be monochromatic, since the partial coloring satisfies $\Pi'$. Therefore, there exists a common color that the children of $v$ can pick that would turn $v$ into a monochromatic node. This way, we can turn all mixed children of $u$ into monochromatic children and make the node $u$ monochromatic. Once $u$ is monochromatic, it can be colored with two colors and therefore it can always pick a color that is not taken by its parent. In the special case when $u$ is the root of the tree, it might not be possible to turn $u$ into a monochromatic node. In this case, we need to make sure that two of the children of $u$ have the same color, such that $u$ can pick the third color. Such a solution has a mending radius of $3$.
\end{proof}

In Appendix~\ref{app:restrictions:rooted_trees}, we will extend this result to $\Delta$-coloring of $\Delta$-regular trees using the same definition for monochromatic and mixed nodes. We will further show that there exist $O(\log^* n)$-solvable edge-checkable problems that cannot be restricted using monochromatic and mixed nodes to be efficiently mendable. 

\section{Landscape of mendability}\label{sec:landscape}

In this section, we analyze the structure of the landscape of mendability in three settings: (1)~cycles and paths, (2)~trees, and (3)~general bounded-degree graphs. We start with a general result that we can use to prove a gap in all of these three settings.

\subsection{A general gap result}

\begin{theorem}\label{thm:lowgap}
	Let $\mathcal{G}_{G_\Delta}$, $\mathcal{G}_{T_\Delta}$, and $\mathcal{G}_C$,  be, respectively, the family general graphs of maximum degree $\Delta$, the family of trees of maximum degree $\Delta$, and cycles (in this case $\Delta=2$). Let $d(n) = n$ if $\Delta \le 2$, and $d(n) = \log n$ otherwise.
	Let $\mathcal{G}$ be one of the above families. There is no LCL problem $\Pi$ defined on $\mathcal{G}$ that has mending radius between $\omega(1)$ and $o(d(n))$.
\end{theorem}
\begin{proof}
	We prove the claim by showing that a mending complexity of $f(n) = \omega(1)$ implies a mending complexity of $\Omega(d(n))$. In particular, we show that for any $n$, we can construct an instance where the mending complexity is $\Omega(d(n))$. Given $n$, we construct the instance as follows. Let $I_N$ be the instance of size $N$ where there is a node $v$ that requires $f(N) = \Omega(d(n))$ to be mended. Such an instance must exist, since by assumption $f(N) = \omega(1)$. Starting from $I_N$, we now create a smaller instance of size at most $n$ having mending complexity $\Omega(d(n))$. The new instance $I'$ is defined as follows. Consider the largest radius $r$ such that the graph induced by all nodes at distance at most $r$ from $v$ contains at most $n' = c n$ nodes, for some small enough constant $c > 0$. If $\Delta=2$, then note that $r = \Omega(n)$, while if $\Delta> 2$ then $r = \Omega(\log_\Delta n) = \Omega(\log n)$. Hence $r = \Omega(d(n))$. Assign $\bot$ as output for the nodes at distance exactly $r$ from $v$, and let $w$ be one of such nodes. We then create a new instance $I''$ of exactly $n$ nodes by adding $n-n'$ nodes as follows. If $\mathcal{G}$ is the family of cycles we add a path of $n-n'$ nodes such that its endpoints are connected to the two endpoints of $I'$. Otherwise, we connect an arbitrary tree of $n - n'$ nodes to $w$ such that the obtained graph contains exactly $n$ nodes and has still maximum degree $\Delta$. Note that the obtained graph is in the same family of the original graph.
	
	We now argue that the mending radius on $I''$ is at least $\min\{r,f(N)\} - t = \Omega(d(n))$, where $t = O(1)$ is the checkability radius of $\Pi$. Assume for a contradiction that the mending radius is $\ell < \min\{r,f(N)\} - t$. Then, by Lemma \ref{lem:normalizedmending} we could have mended node $v$ in the same way on $I_N$, since the neighborhood of $v$ at distance $\ell+t$ is the same on both $I_N$ and $I''$. This contradicts the fact that mending $I_N$ requires $f(N)$.
\end{proof}

\subsection{Landscape of mendability in cycles}\label{chap:cycles}

Now the case of cycles and paths is fully understood. Theorem \ref{thm:lowgap} implies the following corollaries:

\begin{corollary}
	There is no LCL problem with mending radius between $\omega(1)$ and $o(n)$ on cycles.
\end{corollary}
\begin{corollary}
	There is no LCL problem with mending radius between $\omega(1)$ and $o(n)$ on paths.
\end{corollary}

That is, there are only two possible classes: $O(1)$-mendable problems and $\Theta(n)$-mendable problems.

\subsection{Landscape of mendability in trees}

In the case of trees, Theorem \ref{thm:lowgap} implies a gap between $\omega(1)$ and $o(\log n)$:
\begin{corollary}
	There is no LCL problem with mending radius between $\omega(1)$ and $o(\log n)$ on trees.
\end{corollary}

One cannot make the gap wider: there are problems that are $\Theta(\log n)$-mendable, for example, $\Delta$-coloring \cite{panconesi95delta}. However, we can prove another gap above $\Theta(\log n)$:
\begin{theorem}\label{thm:tree-gap-logn-n}
	There is no LCL problem $\Pi$ with mending radius between $\omega(\log n)$ and $o(n)$ on trees.
\end{theorem}

\begin{proof}[Proof sketch.]
	We construct a recursive mending algorithm based on a modified version of the rake and compress decomposition of Miller and Reif~\cite{Miller1985}. Instead of compressing all paths, we only compress sufficiently long paths.
	
	The mending proceeds by recursively pushing the $\bot$ labels down along the decomposition into the compress paths of lower layers. These paths form subinstances that can be mended by changing labels inside at most constant distance along the path. This step relies on the $o(n)$-mendability of $\Pi$. 
	
	The mending algorithm finishes after $O(\log n)$ recursions and all modified labels are within distance $O(\log n)$ of the initial node to be mended.
\end{proof}

The full proof is presented in Appendix~\ref{app:tree-gap-proof}.

\subsection{Landscape of mendability in general graphs}
The mendability landscape on general graphs looks different than the one on cycles and trees. In fact, while for the latter cases we have got two wide gaps, it seems that in the case of general graphs the landscape of mendability is denser.
In this section, we will show that, in general graphs, 
\begin{itemize}[noitemsep]
	\item there \emph{is} a gap in the mendability landscape between $\omega(1)$ and $o(\log n)$, and
	\item there \emph{is not} a gap between $\omega(\log n)$ and $o(n)$. 
\end{itemize} 
As a proof of concept of the second point, we will present a problem that has mending radius $\Theta(\sqrt{n})$, but we believe that the landscape is much denser. The first point follows directly as a corollary of Theorem \ref{thm:lowgap}.

\begin{corollary}
	There are no LCL problems with mending radius between $\omega(1)$ and $o(\log n)$ on general graphs.
\end{corollary}

In Appendix~\ref{app:sqrtn-general} we prove the following result that shows that the landscape of mendability in general graphs is more diverse than the landscape in trees.
\begin{restatable}{theorem}{thmreplacesqrtn}\label{thm:replace-sqrtn-general}
	There are LCL problems that are $\Theta(\sqrt{n})$-mendable on general bounded-degree graphs. 
\end{restatable}
This is merely one example; we conjecture that there are $\Theta(n^\alpha)$-mendable problems at least for all rational numbers $0 < \alpha < 1$.

\subsection{From sublinear mendability to logarithmic solvability}\label{sec:corollaries}

We now give an example of a result that we can easily derive by putting together results from Sections \ref{sec:mendability_to_solvability} and \ref{sec:landscape}:
\begin{corollary}
	Let $\Pi$ be an LCL problem defined on trees with $o(n)$ mending radius. Then $\Pi$ can be solved in $O(\log n)$ rounds in the LOCAL model. 
\end{corollary}
\begin{proof}
	By \autoref{thm:tree-gap-logn-n}, on trees, a mending radius of $o(n)$ implies a mending radius of $O(\log n)$. We prove that a mending radius of $O(\log n)$ implies an algorithm running in $\polylog n$ rounds. Due to a known gap in the complexity hierarchy of LCLs on trees, this implies an algorithm running in $O(\log n)$~\cite[Theorem 3.21]{Chang2019}.
	
	Let $k = O(\log n)$ be the mending radius of $\Pi$, and $r$ be its checkability radius.
	We start by computing a $(c,d)$-network decomposition of $G^{2k+r+1}$, for $c=d=O(\log n)$, that can be done in $\polylog n$ time~\cite{Rozhon2019}. We then apply \autoref{thm:mending_to_nd} to solve $\Pi$ in $\polylog n$ rounds.
\end{proof}

\section*{Acknowledgments}
This project has received funding from the European Union's Horizon 2020 research and innovation programme under the Marie Sk{\l}odowska-Curie grant agreement No 840605. This work was supported in part by the Academy of Finland, Grants 314888 and 333837. The authors would also like to thank David Harris and the  anonymous reviewers for their very helpful comments and feedback on previous versions of this work.

\bibliography{bibliography}

\newpage
\appendix

\section{\{1,3,4\}-Orientation in two-dimensional grids}\label{app:134orientation}
In this section, we will consider the $\{1,3,4\}$-orientation problem in two-dimensional grids. This problem is defined as follows: given a two-dimensional grid, orient the edges of the grid such that each node $v$ of the grid has $\operatorname{indegree}(v)\in \{1,3,4\}$. 
In other words, we forbid the nodes to only have outgoing edges or to have two incoming and two outgoing edges. 

In the following lemma, we will show that this problem is $2$-mendable.

\begin{lemma}\label{lem:mending_134orientation}
	The $\{1,3,4\}$-orientation problem in two-dimensional grids is $2$-mendable.
\end{lemma}
\begin{proof}
	Assume that we have a partial solution with only one hole, i.e. the edges around one node $v$ are not oriented. Assume that there is no feasible orientation of these edges given the partial solution. We will show that it is possible to flip edges inside a $3\times 3$ substructure centered around $v$ such that the resulting orientation is a valid solution. Since the corners of the $3\times 3$ substructure are at distance $2$ from $v$, the problem is $2$-mendable.
	
	We first observe some useful properties of valid orientations around any single fixed node. 
	\begin{itemize}
		\item We can always flip one edge of a node with indegree $4$, as the node will have indegree $3$.
		\item If we flip any two edges of a node with indegree $1$, the node will have indegree $3$ or $1$ (depending on whether we flip two incoming edges, or one incoming and one outgoing edge). Equivalently we can always flip any two edges of a node with indegree 3.
		\item If a node has indegree 0 or 2, flipping an edge gives indegree 1 or 3.
	\end{itemize}
	
	Assume that the edges around $v$ are oriented such that we have a forbidden structure (either $0$ or $2$ incoming edges). We will consider three possible cases for the $3\times 3$ substructure of nodes with $v$ in the center.
	
	Consider first the case where $v$ has a neighbor of indegree $4$. In this case, we can flip the outgoing edge between $v$ and this node: now $v$ has indegree $1$ or $3$ and the other node still has a valid orientation.
	
	In the second case, we exclude the previous case and assume that there is a corner node in the $3\times 3$ substructure with indegree $4$. We can flip both edges on a path from $v$ to this corner node in order to reach a valid orientation. Note that for the node with indegree $4$ and for $v$ we only flip one edge, while we flip two edges of the intermediate node with indegree $1$ or $3$.
	
	The remaining case is when all nodes in the $3\times 3$ substructure, except for $v$, have indegree $1$ or $3$. Consider the cycles spanned by all four $2\times 2$ substructures. Note that we can flip the edges along any cycle such that nodes with indegree $1$ and $3$ are still valid. This is because we flip two edges for each node in the cycle. Assume that there exists a cycle in which $v$ has two outgoing edges. If there is no such cycle, we can first flip all edges in one of the incident $4$-cycles in order to get a cycle with the above property. We can then flip all edges along this cycle with two outgoing edges at $v$ and will either reach a valid orientation with $\operatorname{indeg}(v)=0$, or we will receive a setting with $\operatorname{indeg}(v)=2$. In the latter case, we need to repeat the above step one more time in order to reach a valid configuration.
\end{proof}

\paragraph{An alternative computer-assisted proof.}

We can also use a computational approach to prove Lemma~\ref{lem:mending_134orientation}. We start by assuming that the mending radius of a problem is small and implement a script that checks all possible label assignments in a small neighborhood. If a valid assignment of labels exists for each input configuration, then the problem is mendable at our assumed radius. For the $\{1,3,4\}$-orientation, we will mend a node $v$ following the steps
\begin{enumerate}
	\item Remove all edges in the $3\times 3$ substructure of nodes centered at $v$.
	\item Check the number of incoming edges at each node in this substructure. 
	\item Add directed edges in a brute force manner to the substructure and check if there exists a valid orientation of edges.
\end{enumerate}

In this example, we can check all cases and verify that there exists a solution for any initial assignment of incoming edges in the $3\times 3$ substructure. We first enumerate all possible input configurations: the corner nodes can have up to two incoming edges, the node $v$ has no incoming edges, and the remaining nodes can have at most one incoming edge. This results in $1296$ possible input configurations. For each of the configurations, we need to determine whether there is a valid orientation of the edges by checking at most $2^{12}$ possible edge orientations. We implemented a Python program that checks all possible edge orientations, and it found a valid assignment for each of the input configurations.\footnote{The code is publicly available, but to keep the submission anonymous, we have made the files available in the following shared Dropbox folder: \url{https://www.dropbox.com/sh/857xfi0dtininli/AACFDLOm_McykEMkea9xw-jAa?dl=0}} This computational approach hence also shows that the $\{1,3,4\}$-orientation problem is $2$-mendable.

We expect that a similar computational approach (possibly with a more efficient implementation that only checks non-symmetric cases and uses e.g.\ modern SAT solvers as a subroutine to the right patch for each configuration) can be applied in the study of many other problems, especially in highly structured settings such as grids and trees.

\newpage

\section{Paths, Cycles and Trees}\label{app:paths,cycles,trees}
In Section~\ref{sec:solvability_to_mendability}, we observed that not all efficiently solvable problems are also efficiently mendable and showed that under certain restrictions of a problem instance, LCLs on directed cycles without input can be made mendable. In this section, we will consider analogous results for LCLs on undirected paths and cycles and rooted trees. We will start by recalling the automata-theoretic view for studying LCLs~\cite{chang21automata-theoretic,Brandt2017}.

\subsection{Diagram representation of LCLs}\label{app:intro_automata}

We start by considering a node-edge-checkable LCL problem $\Pi$ that is defined as a tuple $\Pi=\{\Gamma, C_{\Edge}, C_{\Node}, C_{\Start},C_{\End}\}$. Thereby $\Gamma$ is the set of possible output labels, $C_{\Edge}$ and $C_{\Node} \subseteq \Gamma\times\Gamma$ the edge and the node constraints, and $C_{\Start}$ and $C_{\End} \subseteq \Gamma$ the start and the end constraints. In the case of undirected paths and cycles, we will consider \emph{symmetric} LCLs. These are LCLs where $C_{\Start}=C_{\End}$ and where the set of edge constraints as well as the set of node constraints are symmetric relations.  

We can represent an LCL as a diagram $D$ as follows: the nodes of $D$ correspond to the edge constraints $C_{\Edge}$. Between two states $(a,b)$ and $(c,d)$, we draw an edge if $(b,c)\in C_{\Node}$. We further say that a node is a starting state if its first state is in $C_{\Start}$ and an ending state if its last state is in $C_{\End}$.

In order to classify LCL problems defined on diagrams, we need to differentiate between different states that the nodes of a diagram can have:

\begin{definition}[repeatable state]
	A state is called \emph{repeatable} if there exists a directed walk of length $\geq 1$ from a node to itself. 
\end{definition}
\begin{definition}[loop]
	A repeatable state is called a \emph{loop} if there is a walk of length $1$ from a node to itself. 
\end{definition}
\begin{definition}[flexible state]
	A state is called \emph{flexible}, if there exists a flexibility parameter $K$, such that for every constant $k>K$ there is a walk of length $k$ from a node to itself.
\end{definition}
\begin{definition}[mirror-flexible state]
	A state is called \emph{mirror-flexible}, if there exists a flexibility parameter $K$, such that for all $k>K$ there is a walk of length $k$ from a state $(a,b)$ to $(b,a)$.
\end{definition}

In the case of cycles and paths of length $n$, a walk of length $n$ in the diagram representation of the respective LCL problem corresponds to a valid labeling of the nodes of the path or cycle. Different properties of the diagram define different asymptotic running times for solving the LCL problems. If, for example, there exists a state with a self-loop in a directed cycle, then the corresponding problem is $O(1)$-solvable, as all nodes can label their edges using the same states~\cite{Brandt2017}. In the case of undirected path or cycles, a mirror-flexible state with a self-loop is required and sufficient for $O(1)$-solvability~\cite{chang21automata-theoretic}. Note that $O(1)$-mendability is not possible in such diagrams in general. One can however restrict the diagrams to only consist of a mirror flexible state with a loop (or just the self-loop in the directed case) and make the new problem $O(1)$-mendable. 

In the following sections, we will consider similar restrictions for $O(\log^* n)$-solvable LCLs on undirected path and cycles as well as on rooted trees. We will show that for undirected path and cycles there exist similar restrictions to Theorem~\ref{thm:cycle_mend_hard_to_easy} that would make the problems $O(1)$-mendable. The case with rooted trees is trickier. For this case, we first present a generalized restriction to the one in Theorem~\ref{thm:Three-coloring_binary_trees} in order to show that the $\Delta$-coloring problem on $\Delta$-regular trees can be made $O(1)$-mendable. We then show that not all problems on directed trees can have such a natural restriction.

\subsection{LCLs on undirected paths}\label{app:restrictions:undirected_paths}
For undirected paths, the situation is slightly different from the directed case. Here, we need to consider symmetric LCL problems. Such problems on paths or cycles are $O(\log^* n)$-solvable if there exists a mirror-flexible state in the diagram representation~\cite{chang21automata-theoretic}. Note that symmetry and a mirror-flexible state are not yet sufficient to make a problem efficiently mendable. This is because two parts of a path can be colored using colors from two different color sets, similar to the proof of Theorem~\ref{thm:path_logstar_linear}. In order to make the problem efficiently mendable, we need to restrict the problem to only one connected component with a mirror-flexible state. The precise restriction is given in the following theorem:

\begin{theorem}\label{thm:restriction_for_mendability_on_paths}
	Suppose $\Pi$ is a symmetric LCL problem on undirected paths without input. If $\Pi$ is $O(\log^* n)$-solvable, we can define a new LCL problem $\Pi'$ with the same time complexity, such that a solution for $\Pi'$ is also a solution for $\Pi$, and $\Pi'$ is $O(1)$-mendable.  
\end{theorem}
\begin{proof}
	Observe first that if $\Pi$ is $O(1)$-solvable, then there exists a flexible state with a self-loop that is a start and an end state. Therefore, we can restrict the $\Pi'$ to be this state and achieve an $O(1)$-mendable problem. 
	
	For the general case, the following restrictions to the graph representation are required in order to achieve efficient mendability of $\Pi'$:
	\begin{itemize}
		\item The considered states must form a connected component.
		\item All repeatable states in this component form a strongly connected component.
		\item There exists a flexible state inside this component.
	\end{itemize}
	These conditions imply that a flexible state will always be inside the strongly connected component of repeatable states, as it is a repeatable state as well. This, on the other hand, implies that all states inside the strongly connected component are flexible, since all states in the strongly connected component can be reached from the flexible state in a constant number of steps. Finally, we consider symmetric LCLs and therefore all symmetric states will be inside the strongly connected component. Therefore, all flexible states must also be mirror-flexible. Let $k$ be the maximum flexibility of these nodes. Note that such a restriction is always possible, since there always exists a repeatable state in the graph representation that either has a self-loop (trivial case) or belongs to a larger connected component of the graph.
	
	Observe that our restriction is $O(\log^* n)$-solvable as the remaining problem is still symmetric and preserves the mirror-flexible states. Next, we will show that the problem is $O(1)$-mendable. We first need to address the non-repeatable nodes that can be contained in the connected component. The size of connected components of non-repeatable nodes is bounded by the size of the connected component, which we denote as a constant $q$. Moreover, we will assume that the start and the end states are in such components, as the states could be removed from the graph otherwise. 
	
	Consider some node $u$ that is an inner node of the path, i.e., it has distance $> k+q+1$ to each of the endpoints. Assume further that the $k+q$-neighborhood of $u$ is already colored. In order to mend the solution of $u$, we will first delete the solution of all nodes at a radius $k+q$. 
	Next, if the color of the node at distance $k+q+1$ is from the strongly connected component, we fix the coloring at radius $k+1,\ldots,k+q$ using the labels from the strongly connected component in $\Pi'$. If the color is inside a non-repeatable component, we use the labels from the non-repeatable component for coloring, until we reach the labels from the strongly connected component. Note that this situation can only appear if a node at distance $k+q+1$ is close to one of the endpoints. It remains to fix the colors on the path between the neighbors at distance $k+1$ from $u$. We call the nodes at distance $k+1$ $v_1$ and $v_2$, respectively. We can find a walk in the graph of length $2k$ that starts in the state of $v_1$ and ends in the state of $v_2$ and use the labels along this walk as colors on the path between $v_1$ and $v_2$.
	
	We now will consider some special cases. We start with the case where not all neighbors of $u$ are colored. In this case, it might not always be possible for $u$ to pick a consistent color: assume that one node at distance $k+q+1$ has a color from a non-repeatable component. We call this node $w$. If there exists an uncolored node between $u$ and $w$, there does not necessarily exist a walk of length $k+q+1$ in the connected component starting in the state of $w$. In such a case, the other uncolored node will be able to fix the coloring.  
	The second case is when $u$ is close to one of the endpoints. Then, $u$ can delete the solution at a radius at $k+q$, and make sure that the end node satisfies the start and the end constraints.
\end{proof}

\subsection{LCLs on undirected cycles}\label{app:restrictions:undirected_cycles}
Undirected cycles and paths are treated equally when solvability of LCLs is considered. When we consider mending, this property changes. For paths, we could assume that there exist connected components in the diagram representation of $\Pi'$ that only consist of non-repeatable states. In the case of cycles, the start and the end states cannot be inside the non-repeatable component. This is because a mending algorithm cannot rule out that there will be only one start and end point along the cycle without seeing the whole cycle. We therefore need to restrict the considered graph representation to only consist of the strongly connected component.    

\begin{theorem}\label{thm:restriction_for_mendability_on_cycles}
	Suppose $\Pi$ is a symmetric LCL problem on undirected cycles without input. If $\Pi$ is $O(\log^* n)$-solvable, we can define a new LCL problem $\Pi'$ with the same time complexity, such that a solution for $\Pi'$ is also a solution for $\Pi$, and $\Pi'$ is $O(1)$-mendable.
\end{theorem}
\begin{proof}
	Observe first that the case where the problem is $O(1)$-solvable is equivalent to the proof of Theorem~\ref{thm:restriction_for_mendability_on_paths}.
	
	In order to achieve efficient mendability on cycles, we need to restrict the graph representation of $\Pi$ to $\Pi'$ as follows:
	\begin{itemize}[noitemsep]
		\item The considered states form a strongly connected component.
		\item There exists a flexible state inside this component.
	\end{itemize}
	Note that this restriction is always possible, since there will always be a connected component in the graph representation that contains a mirror-flexible state, otherwise, the problem would not be $O(\log^*n)$-solvable. Add all states that lie on any path between the two mirror-flexible nodes. Note that all these states are repeatable and form a strongly connected component. As in the proof of Theorem~\ref{thm:restriction_for_mendability_on_paths}, we can conclude that all states in the strongly connected component are mirror-flexible.
	
	It remains to show that all such problem instances are constant radius mendable. Let $k$ be the maximum flexibility in the strongly connected component. Further, let $q$ be the number of nodes in the considered strongly connected component. In order to mend the solution, node $u$ needs to remove the solution of all nodes in its $k+q$ neighborhood. It then can use labels from a walk of length $2(k+q)$ in the graph in order to complete the coloring between the two neighbors at distance $k+q+1$.
\end{proof}

\subsection{LCLs on rooted trees}\label{app:restrictions:rooted_trees}
In~\cite{chang21automata-theoretic}, it was shown how to use diagram representation in order to classify edge-checkable problems in rooted trees. The asymptotic round complexity of solvability has been shown to be the same on rooted trees and directed paths. We cannot derive a similarly strong statement for mendability. In fact, already for the problem of $3$-coloring binary trees, we can show that no restriction of the admissible configurations can make the problem efficiently mendable.

\begin{figure}
	\centering
	\includegraphics[width=0.9\textwidth]{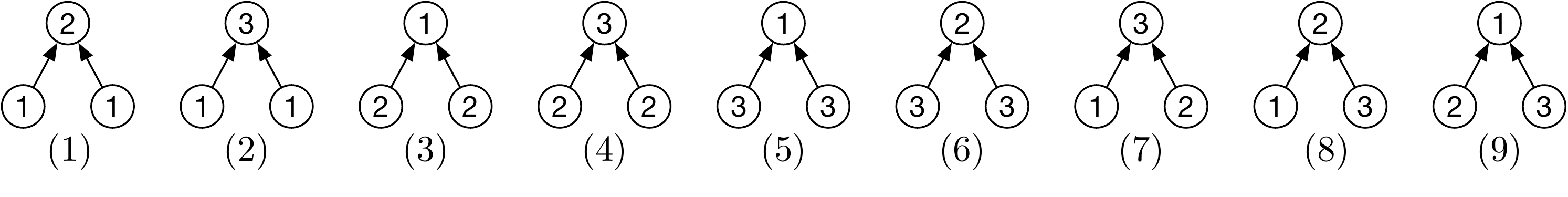}
	\caption{Possible configurations of the $3$-coloring problem on binary trees in diagram representation with colors $1,2$ and $3$.}\label{fig:three_coloring_possible_configs}
\end{figure}

\begin{theorem}
	Consider the diagram representation of the $3$-coloring problem on binary trees. All restrictions of the configurations of this problem are $\Omega(\log n)$-mendable.
\end{theorem}
\begin{proof}
	Figure~\ref{fig:three_coloring_possible_configs} presents the possible configurations for the $3$-coloring problem on binary trees. Note that we cannot keep all configurations with mixed nodes, as the three configurations can form a rigid structure that can be arbitrarily large. It can be shown that for each subset of configurations there exists a counterexample that has a mending radius of $\Omega(\log n)$. We will restrict our analysis to two selected counterexamples, other counterexamples can be derived in a similar way.
	
	Assume that we restrict the configurations by removing configuration $(9)$ from the list. This means that no node in the tree is allowed to have children of colors $2$ and $3$. We can use this restriction to construct a rigid substructure that can be repeatedly attached to its leaves. This substructure is represented in Figure~\ref{fig:three_coloring_counterexample_1}. Note that the uncolored node has two children of colors $2$ and $3$. Therefore, one of the children has to be recolored. The child of color $2$ has children of different colors such that the recoloring step is propagated to its children. Changing the child of color $1$ to color $2$ would create two siblings of colors $2$ and $3$. Therefore, the recoloring step will always be propagated into the substructure $S$. Note that $S$ is a rigid substructure, i.e. if the root of $S$ has to be recolored, all nodes along one root to leaf path have to be recolored as well. Since $S$ is attached to each leaf node of $S$ recursively, the recoloring propagates all the way to the leaves of the tree. Therefore, mending requires $\Omega(\log n)$ radius in this case. A similar counterexample can be derived if two of the configurations $(7), (8)$ and $(9)$ are removed. 
	
	We will next consider the case when the problem $\Pi'$ is restricted to the configurations $(1),(2),(3)$ and $(6)$. The rigid structure for this example is presented in Figure~\ref{fig:three_coloring_counterexample_2}. Here, the fact that the nodes cannot have children of two different colors is used in order to restrict the number of color choices for the children. Since the children of the uncolored nodes have two different colors, one of them has to be recolored. Note that recoloring the node of color $2$ to color $1$ would require recoloring the roots of $S$. The rigid structure $S$ makes sure that both children of the uncolored node have to be recolored together with the root. For one of the children, the color is fixed by its children, and therefore the recoloring would propagate down to the grandchildren. Therefore, any color change propagates down to the leaves of $S$. Counterexamples for other choices of configurations that do not contain the configurations $(7), (8)$ and $(9)$ can be derived in a similar way.
\end{proof}

\begin{figure}[tbh]
	\centering
	\includegraphics[width=0.9\textwidth]{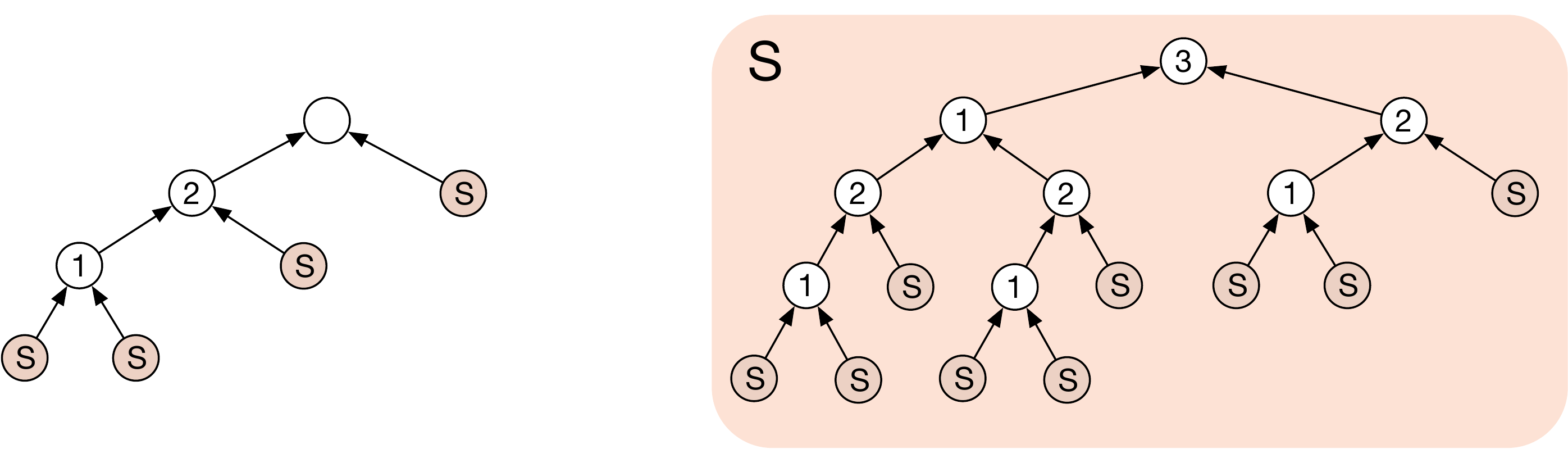}
	\caption{Counterexample for the 3-coloring problem on binary trees, where configuration $(9)$ from Figure~\ref{fig:three_coloring_possible_configs} is missing. On the left side, the construction of the tree with the uncolored node is presented. The leaves $S$ are replaced by the recursive substructure from the right side. }\label{fig:three_coloring_counterexample_1}
\end{figure}

\begin{figure}[tbh]
	\centering
	\includegraphics[width=0.9\textwidth]{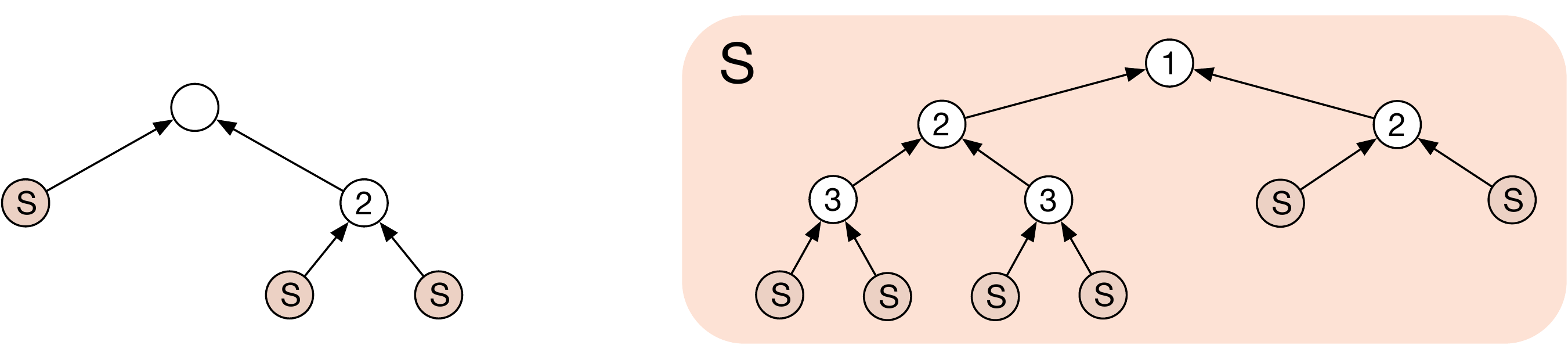}
	\caption{Counterexample for the 3-coloring problem on binary trees, where only configurations $(1),(2),(3)$ and $(6)$ from Figure~\ref{fig:three_coloring_possible_configs} are considered. On the left side, the construction of the tree with some uncolored inner node is presented. The leaves $S$ are replaced by the recursive substructure from the right side. }\label{fig:three_coloring_counterexample_2}
\end{figure}

In Theorem~\ref{thm:Three-coloring_binary_trees}, we showed that by differentiating between monochromatic and mixed nodes in the possible configurations, we can make the new problem $O(1)$-mendable. In the following theorem, we generalize this result to the $\Delta$-coloring problem on $\Delta$-regular rooted trees and show that this problem is constant radius mendable, even if mixed nodes form connected components in the tree.

\begin{theorem}\label{thm:Delta-coloring_Delta-regular_trees}
	Given the $\Delta$-coloring problem $\Pi$ on $\Delta$-regular rooted trees without input, where $\Delta > 2$, we can define a new LCL problem $\Pi'$ with the same time complexity, such that a solution for $\Pi'$ is also a solution for $\Pi$, and $\Pi'$ is $O(1)$-mendable. 
\end{theorem}
\begin{proof}
	Note first that the $\Delta$-coloring problem on $\Delta$-regular rooted trees is not $o(\log n)$-mendable in general: if all $\Delta-1$ children of a node pick different colors, this node can only be colored with one remaining color. This means, that the leaf nodes can define a unique coloring of the whole tree. For mending, we can assume that the root node has $\Delta$ children, all colored in different colors, and that the color of a child node is fixed at a leaf level. Then, in order to mend the color of the root, we would need to undo the coloring of all nodes in one of the subtrees, thus requiring a mending radius of $\Omega(\log n)$.
	
	In order to prevent such cases, we need to restrict the problem. We will therefore differentiate between two kinds of nodes. We call nodes whose children are colored and all have the same color \emph{monochromatic}, and nodes that have children of at least two different colors \emph{mixed}. WLOG all leaf nodes are assumed to be monochromatic. We define a new problem $\Pi'$ by restricting the original problem $\Pi$ as follows: 
	\begin{itemize}
		\item Assume that a set of mixed nodes forms a connected component in the tree. Then, the tree height of the connected component must be smaller than a constant $k\geq 1$.  
	\end{itemize}
	Note that under such a restriction, $\Pi'$ can still be solved in $O(\log^* n)$ rounds. This solution is based on a $\Delta+1$-coloring algorithm and a subsequent "shift down" strategy, where, in one round, all nodes pick the color of their parent node. Thus, the algorithm always computes a solution with only monochromatic nodes.
	
	We will next show that $\Pi'$ is $O(1)$-mendable. The main idea of this proof is based on the fact that nodes in connected components of mixed nodes can be turned into monochromatic nodes by recoloring them as follows: find the lowest mixed node inside the connected component whose children are all monochromatic. Call this node $v$. Observe that each of the monochromatic children can be colored with at least $\Delta-1$ different colors, as their children (the grandchildren of $v$) can only block one color. The $\Delta-1$ children of $v$ can therefore pick a common color, thus making $v$ monochromatic. The same procedure can be continued in a bottom to top fashion for all nodes in the connected component of mixed nodes. 
	
	In order to mend the solution around an uncolored node $u$, we will apply the above strategy to each mixed child of $u$. We first consider the case where neither $u$ nor the parent of $u$ are the root. If the parent of $u$ is a mixed node, or if it is uncolored, we can turn all children of $u$ into monochromatic nodes and color them with the same color. Then, we pick a color for $u$ that is not in conflict with the color of its parent. Since $u$ is a monochromatic node in this strategy and since we did not introduce any new mixed nodes, the new (partial) solution is still valid for $\Pi'$.
	
	If the parent of $u$ is a colored monochromatic node, we need to make sure that we do not turn it into a mixed node, as this might increase the size of its connected component of mixed nodes. In this case, we will turn $u$ into a mixed node instead. Since $u$ only has monochromatic children after recoloring its subtrees, we can color $u$ with the same color as its siblings. In the last step, we need to find a valid coloring for the children of $u$. This is always possible, since the children of $u$ are monochromatic and therefore can choose one of the $\Delta-2 > 1$ possible colors. This strategy can turn $u$ into a mixed node. Since the parent and the children of $u$ are all monochromatic, $u$ is an isolated mixed node and therefore the new (partial) solution is valid for $\Pi'$.   
	
	Two special cases remain in our analysis. In the first case, the parent of $u$ is the root. Here, the siblings of $u$ can be mixed nodes as well. Turning the root into a mixed node will increase the size of the connected component of mixed nodes by $1$. In order to prevent this situation, we need to apply our recoloring strategy to $u$'s own subtree as well as the subtrees of its siblings. This way, all siblings will be turned into monochromatic nodes. Finally, we can pick the color of the siblings such that they have a different color from the root.
	In the second case, $u$ itself is the root of the tree. Then, it might not be possible to turn the root into a monochromatic node. However, the root can turn all its children into monochromatic nodes using a strategy similar to the above. The root can then choose an arbitrary color and recolor the children. In these cases, the root might be turned into an isolated mixed node. Such a solution still satisfies $\Pi'$.
	
	Observe that, in the worst case, we need to apply the recoloring strategy to a node and its siblings. This would require a mending radius of $k+3$.
\end{proof}

While such a restriction to monochromatic and mixed nodes is natural for edge-checkable problems on directed trees, not all edge-checkable problems can be made $O(1)$-mendable this way. The next theorem shows a construction, where all nodes in a directed tree are monochromatic, but where the solution is not $O(1)$-mendable, even if we allow the mending algorithm to turn all monochromatic nodes into mixed ones.    

\begin{figure}
	\centering
	\includegraphics[width=0.6\textwidth]{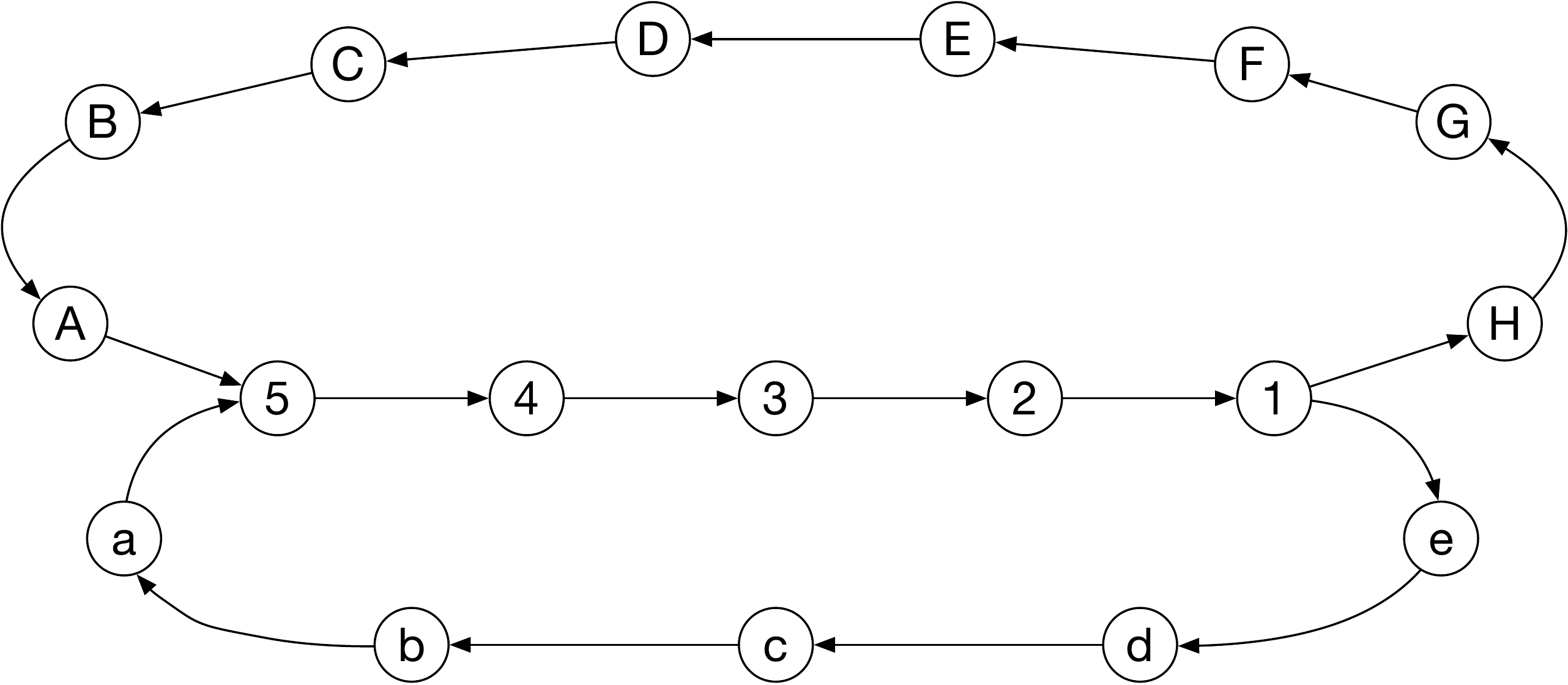}
	\caption{Diagram representation of the counterexample from Theorem~\ref{thm:coloring_trees_impossibility_example}.}\label{fig:coloring_impossibility_def}
\end{figure}

\begin{theorem}\label{thm:coloring_trees_impossibility_example}
	There exist a symmetric LCL $\Pi$ on rooted trees without input, for which any restriction using monochromatic and mixed nodes to a new LCL problem $\Pi'$ with the same time complexity results in $\Omega(\log n)$-mendability. 
\end{theorem}
\begin{proof}
	We will start by defining the considered symmetric LCL problem on trees. This problem is represented by two overlapping directed cycles in the diagram representation as presented in Figure~\ref{fig:coloring_impossibility_def}. The first cycle $c_1$ consists of labels $A,B,C,D,E,F,G,H$ and the second cycle, $c_2$, of $a,b,c,d,e$. Both cycles overlap in the states $1,2,3,4,5$. Note that the two constructed cycles of length $13$ and $10$ are coprime and that all states in the graph are therefore flexible states. That is, the presented problem is $O(\log^* n)$-solvable. 
	The idea of this proof is to show that under any definition of monochromatic and mixed nodes, it is possible to construct a graph consisting only of monochromatic nodes which is $\Omega(\log n)$-mendable.
	
	We will construct a counterexample on a rooted binary tree using the following substructure $S$: Let the root of $S$ have color $4$. We will split the subtree under the root into a left subtree of height $13$ that will be colored using labels from $c_1$, and a right subtree of height $10$ that we will be colored with labels from $c_2$. Note that the leaves of $S$ will be colored with label $4$. The idea is then to attach the same substructure $S$ to each of the leaves of $S$ recursively.
	
	In the considered counterexample, the root of the tree is assumed to be the only uncolored node that will be mended. Let the left child of the root have color $4$. Attach $S$ recursively at this child until the leaves of the tree are reached. Further, let the right child of the root have color $3$ and its grandchildren color $4$. Apply the same strategy to the right grandchildren of the root. Observe first that all colored nodes in the constructed example are monochromatic nodes according to the definition of monochromatic and mixed nodes in Theorem~\ref{thm:Three-coloring_binary_trees}. Figure~\ref{fig:coloring_impossibility_counterexample} visualizes this construction and the substructure $S$.
	
	Observe further that the presented substructure $S$ is rigid, i.e. changing the label at the root of the substructure will imply that we have to change the labels of some of its leaves as well. Assume that we replace label $4$ with any other admissible label. The closest descendants with label $4$ lie either at distance $10$ or at distance $13$ from the root. Assume that we replace the label $4$ with label $5$. From label $5$, we can reach label $4$ within $9, 12, 19, 25, 29, \ldots$ steps. In particular, we will not be able to match the labels of the leaves of $S$ which would occur within $10$ and $13$ steps. That is, if we assign label $5$ to the root, we would have to undo the colors of all nodes in $S$ and also change the colors of all leaves. The same holds for almost all other labels besides $B$ and $1$. Note that from label $B$, we would be able to reach label $4$ within $10$, but not within $13$ steps, and from label $1$ we would reach label $4$ within $13$, but not $10$ steps. However, there is no other label, from which we would be able to reach $4$ within $10$ and $13$ steps. This means, that, no matter which other label we use to replace label $4$ in the root of $S$, there will be a subtree in $S$ for which we will have to change the labels of the leaves. The same argument can then be applied to the leaves of $S$ and thus the recoloring would propagate down to the leaves of the tree. 
	
	Note that the solution of the constructed example cannot be extended at the root. In fact, since all labels in the tree are uniquely defined by the colors of the leaves, a mending algorithm would have to undo the coloring in one of the subtrees of the root in order to produce a valid coloring. Such a change in the root would propagate to some leaves of the tree. This example would therefore have a mending radius of $\Omega(\log n)$.
\end{proof}
\begin{figure}[tbh]
	\centering
	\includegraphics[width=0.9\textwidth]{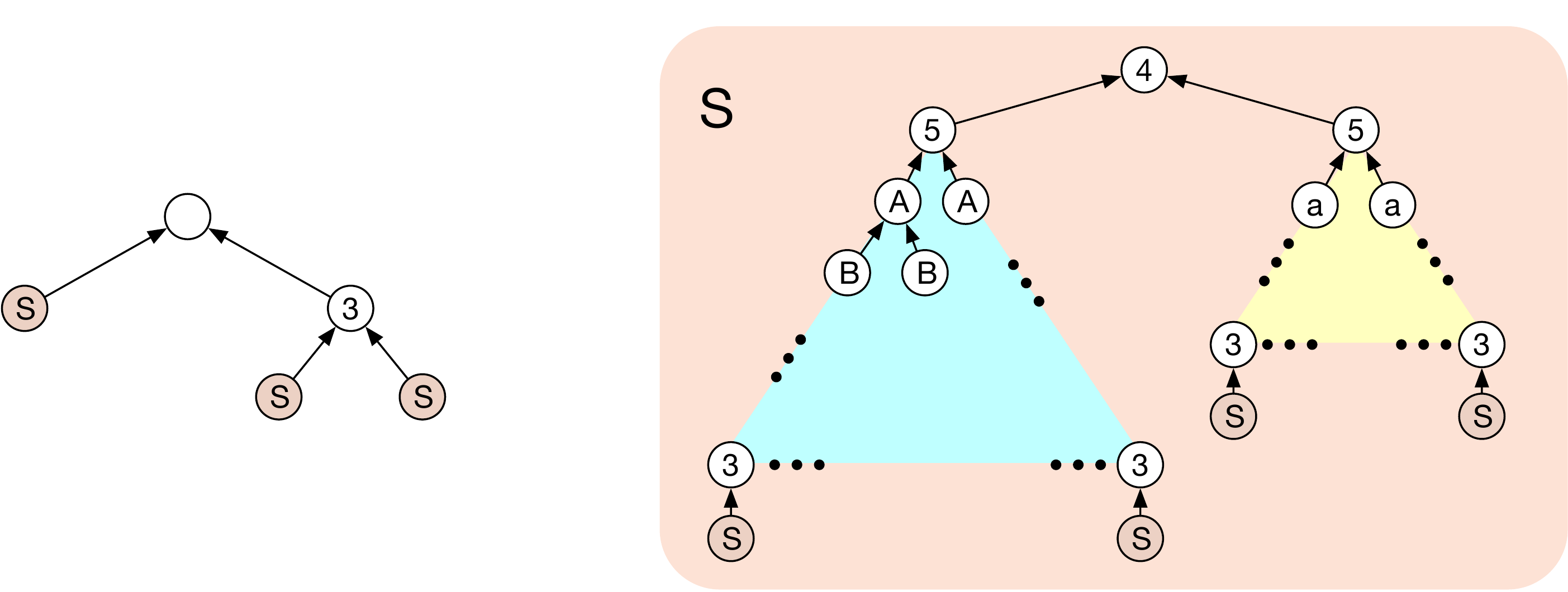}
	\caption{Counterexample from Theorem~\ref{thm:coloring_trees_impossibility_example}. On the left side, the construction of the tree with some uncolored inner node is presented. The leaves $S$ are replaced by the recursive substructure from the right side.}\label{fig:coloring_impossibility_counterexample}
\end{figure}
Observe that the proof of this theorem also holds if we restrict the definition of monochromatic nodes even more: let a node be monochromatic, if the descendants in each layer of its height-$k$ subtree have the same colors. In a height-$2$ subtree the children, as well as the grandchildren, would have the same colors. We can modify the above construction by performing the splitting step at the node of color $3$. This would make sure that the nodes in all height-2 subtrees have the same color and are therefore monochromatic. For larger $k$, we would have to extend the size of the overlapping set of the two cycles, while keeping the property that the cycle length are coprime. 

\newpage

\section{Proof of the complexity gap for trees} \label{app:tree-gap-proof}

In this appendix, we give the full proof of Theorem~\ref{thm:tree-gap-logn-n}.

\begin{proof}[Proof of Theorem~\ref{thm:tree-gap-logn-n}]
	Let $\Pi$ be an arbitrary LCL problem on trees that is $o(n)$-mendable. We prove that it is also $O(\log n)$-mendable.
	
	Let $\mathcal{T} = (V,E)$ be a tree and let $\loutput(v) \in \Gamma \cup \{\bot\}$ be the output assigned to nodes $v \in V$. We describe a procedure that, given a node $v \in V$ such that $\loutput(v) = \bot$, it is able to mend the solution by modifying the output labels of nodes at distance at most $O(\log n)$ from $v$.
	
	We start by assigning to each node $u \in V$ a temporary label $\ell(u) \in \bigcup_{1\le i\le L}\{R_i,C_i\}$, for some $L \in O(\log n)$, by running the following procedure that depends on a constant $k$ to be fixed later. The value $i$ for which $\ell(u) \in \{R_i,C_i\}$ is also called the layer of $u$. This procedure is a modified version of the rake and compress procedure of Miller and Reif~\cite{Miller1985}: In the compress phases, we remove degree-$2$ nodes only if they are inside paths of length at least $2k+1$.
	
	Repeat the following for all $i$ from 1 to $L = O(\log n)$.
	\begin{enumerate}
		\item Remove nodes of degree $\le 1$, assign the label $R_i$ to these nodes.
		\item Remove nodes of degree exactly $2$ that are in paths of length at least $2k+1$. Assign the label $C_i$ to these nodes.
	\end{enumerate}
	
	By choosing a sufficiently large $L$ the obtained graph is always empty: the analysis is the same as for the standard rake and compress, except short paths may require an additional $2k$ rake steps. Since $k$ is a constant, $O(\log n)$ iterations suffice.
	
	Call the nodes that have been assigned label $R_i$, for any $i$, \emph{rake nodes}, and the nodes that have been assigned label $C_i$, for any $i$, \emph{compress nodes}. In the latter case, \emph{compress path} of a node is the path of length at least $2k+1$ in which the node was when it was removed. We prove that the decomposition produced by the modified rake and compress procedure has the following \emph{separation property}:
	\begin{quote}
	Let $v_1$ and $v_2$ be two nodes such that $\{\ell(v_1),\ell(v_2)\} \subseteq \{ C_i, R_i\}$, for some $i$, such that if $\ell(v_1)= \ell(v_2) = C_i$ then $v_1$ and $v_2$ are in different compress paths. 
	Then, in the forest $\mathcal{T}[\leq i]$ induced by the nodes in layers $j \leq i$, the shortest path between $v_1$ and $v_2$ must contain at least one node $v_3$ satisfying $\ell(v_3) = C_j$ for some $j < i$, or $v_1$ and $v_2$ are disconnected.
    \end{quote}
	
	Since every rake node has by definition at most one neighbor at a higher layer, following only rake nodes cannot lead to a different compress path in $\mathcal{T}[\leq i]$. Therefore the path must contain nodes of a compress path. Since $v_1$ and $v_2$ were assumed to be either rake nodes or compress nodes from different paths, this compress node must be from another compress path of a lower layer. 
	
	Let $\mathcal{F}$ be the forest obtained by the following procedure: start with $\mathcal{T}$ and remove all compress nodes that are at distance of at least $k$ from both endpoints of their compress path. Note that this procedure removes at least one node for each compress path. Let $Z$ be the set of removed nodes. Each tree $T$ of $\mathcal{F}$ has diameter $O(\log n)$: By construction, each rake node has at most one neighbor at the same or a higher layer, and each compress node has at most one neighbor at a higher layer. A path starting at a node of the top layer must pass into a lower layer every $k+2$ steps, since by the separation property there is no path from one rake node of any layer to another rake node of the same layer, and no path from a node of a compress path of any layer to a node of a different compress path in the same layer that does not pass through a node of a higher layer. This implies that every path from the top layer of each tree $T$ has length at most $(k+1)L$, and therefore each tree $T$ has radius at most $O(\log n)$.
	
	The mending procedure starts as follows. If $v$ is in $Z$, we skip this first phase. Let $T \in \mathcal{F}$ be the tree containing $v$. Such a tree exists because $v$ is not in $Z$. We start by writing $\bot$ on all the nodes of $Z$ that are neighbors of some node of $T$. We will later remove these temporary labels by mending at those nodes. Since $\Pi$ is mendable, then at least one valid output labeling $s : V \rightarrow \Gamma$ for $\Pi$ in $\mathcal{T}$ must exist. 
	Hence, we can mend $v$ and all nodes in $T$ by assigning $s(u)$ as output label for all nodes $u \in T$, and since $T$ is connected to the rest of the graph only through nodes labeled $\bot$, then the verifier accepts at all nodes of $T$. Since $T$ has diameter $O(\log n)$, this step changes the labeling at distance at most $O(\log n)$ from $v$.
	
	In order to perform this mending, we had to set $\bot$ as output for some nodes that were possibly not labeled $\bot$. Since this is not valid mending procedure, we next  recursively mend at those $\bot$. The nodes at which we have written $\bot$ are special, since the following property holds: all nodes that we still need to mend are inside compress paths, at distance at least $k$ from both endpoints.
	
	We now use the following procedure to mend the nodes where we have written a $\bot$. Let the set of these nodes be $X$. Note that all nodes $u \in X$ are labeled $\ell(u) = C_i$ for possibly different values of $i$. Let $x$ be the largest value, that is, $x = \max\{ i ~|~ \ell(u) = C_i, u \in X \}$. We show a procedure that removes the $\bot$ from all nodes $u$ satisfying $\ell(u) = C_x$ by possibly assigning $\loutput(w) = \bot$ to some nodes $w \in Z$ satisfying $\ell(u) = C_j$, for $j < x$, that are at distance $O(\log n)$ from $u$. By repeating this procedure $O(\log n)$ times we have mended the nodes of all layers, and thus the mending is complete. We will later show that the mending radius is $O(\log n)$.
	
	All nodes $v$ that need to be mended satisfy the following. There is a path $P$ of length at least $2k+1$ of nodes labeled $C_x$ and $v$ is at distance at least $k$ from both endpoints. Let $\mathcal{T} \setminus P$ be the forest obtained by removing the nodes of $P$ from $\mathcal{T}$, and let $\mathcal{S}$ be the set of trees of $\mathcal{T}\setminus P$ connected to the nodes of $P$ that are at distance at most $c$ from $v$, for $c < k$ to be fixed later. 
	
	We start by inserting $\bot$ to the compress paths just as in the first phase of mending. This ensures that the trees in $\mathcal{S}$, after removing the $\bot$-labeled nodes, have again diameter $O(\log n)$.
	Let $S$ be an arbitrary tree in $\mathcal{S}$ and let $u$ be the node of $S$ connected to a node of $P$. Let $T \in \mathcal{F}$ be the tree containing $u$ (recall that $\mathcal{F}$ was the tree obtained by removing the center nodes $Z$ of the compress paths from the original tree $\mathcal{T}$). Remove all nodes labeled $C_x$ from $T$, and from the obtained forest let $T'$ be the tree containing $u$. We write $\bot$ on all the nodes of $Z' \subseteq Z$ that are neighbors of some node of $T'$. 
	Note that by the separation property all nodes of $Z'$ are in layers $j < x$, since higher layer nodes are connected to $T'$ by a path of compress nodes which is cut by $Z$.
	
	It would be easy to now mend the tree $T$ in $O(\log n)$ in the same way as in the first phase of the procedure, except for the nodes of $P$ that the tree is connected to. We prove that, in order to mend $v$, it is enough to modify the labels in the subgraph obtained by considering the nodes of the path $P$ at distance at most $c$ from $v$, for some constant $c$, and all the subtrees starting from them, noting that all these subtrees have now diameter $O(\log n)$.
	
	Intuitively, we proceed as follows: we create a new virtual instance where we replace all these subtrees with ``equivalent'' trees of diameter $O(1)$, where the constant does not depend on $k$. In this way, we obtain an instance consisting of a long path with small trees connected to it. Since $\Pi$ is $o(n)$-mendable, there exists a minimum $k$ such that in all instances where the path has length $2k+1$ the mending radius is at most $c = k - r$ (where $r$ is the checkability radius of $\Pi$). Arguing similarly as in the proof of Theorem~\ref{thm:lowgap}, we can mend all such instances with radius $c$.
	
	The size of the equivalent trees will be independent of $k$ and they satisfy the following property: we can project the output of the mending procedure from the virtual instance to the original one such that the solution near the path $P$ is the same, and that we can also mend the remaining nodes of the subtrees.
	
	We will now describe how to construct the virtual instance. Let $S$ be an arbitrary tree in $\mathcal{S}$ and let $u$ be the node of $S$ connected to a node of $P$. Let $L(S)$ be the set of all possible output labelings for the nodes of $S$ satisfying that the nodes at distance at most $r$ (the checkability radius of $\Pi$) from $u$ have the same labeling as under $\loutput$ on $\mathcal{T}$, and such that all nodes of $S$ satisfy the verifier. For each labeling $l \in L(S)$, let $l_t(l,u)$ be the labeling of nodes at distance at most $t$ from $u$. Let $L_t(S,u) = \{ l_t(l,u) | l \in L(S)\}$. Let $S'$ be the smallest tree (breaking ties arbitrarily) satisfying the following:
	\begin{itemize}
		\item $S'$ contains a node $u'$ such that the radius-$r$ view around $u'$ is isomorphic to the radius-$r$ view of $u$ in $S$, including the partial output labeling.
		\item $L_{2r}(S,u) = L_{2r}(S',u')$, that is, the sets of possible completable labelings of $S$ and $S'$ up to distance $2r$ are the same.
	\end{itemize}
	The tree $S'$ must exist, as the original tree $S$ satisfies the properties. Since the maximum degree, the checkability radius, and the number of labels are constant, and since the constraints cannot depend on identifiers, there are only finitely many different possible values for $L_{2r}(S,u)$. Hence, $S'$ must have a number of nodes that is independent of $n$ or $k$, meaning that it has constant size. Also, $S'$ is computable (by brute force) given $S$.
	The tree $S'$ is equivalent to $S$ in the following sense: if we fix some labeling near $u'$ that can be completed to a valid labeling for all nodes of $S'$, then we can fix the same labeling near $u$ and complete $S$ as well. Let $c(S)$ denote the tree $S'$ computed for $S$ with the process described above.
	
	We now create the following instance: consider $P$ and all the subtrees $S \in \mathcal{S}$ connected to it. For each $S$, replace it with $c(S)$. This gives a path-like instance, where the path has a length of at least $2k+1$, and to each node of the path is attached a tree of constant size. We now run the mending procedure on the node of this instance corresponding to $v$. As observed previously, we can do this by modifying the solution of nodes at distance at most $c \leq k - r$ from $v$. This ensures that the modifications do not affect nodes outside $P$ and $\mathcal{S}$. We now copy the solution written on the nodes of the path and all nodes at distance at most $2r$ from them to the original instance. Then, for all the subtrees $S \in \mathcal{S}$ connected through the path via node $u$, we modify the solution of all the nodes at a distance strictly larger than $2r$ from $u$, such that the verifier accepts the solution on all nodes of $S$. Such a solution must exist by definition of $c(S)$. 
	
	We can apply this procedure recursively to obtain a mending of $\mathcal{T}$. It remains to show that the mending radius is $O(\log n)$.
	On a high level, this holds because whenever we recurse on a node $u$ far away from the original node $v$ that was labeled with a temporary $\bot$, then the path from $v$ to $u$ must go through rake nodes and at most two partial compress paths of descending layers. Therefore the recursion at $u$ will start from a lower layer. 
	
	More formally, consider the second phase of the procedure we have described. We have a node $v$ in some layer $i$, and in order to mend this node we recursively need to mend nodes $u$ in layers $j < i$. We can observe that the distance between $u$ and $v$ must be at most $2k + i - j$. Let $(r_1,\ldots,r_t)$ be an arbitrary sequence of the layers of the nodes at which we recurse, that is, we want to mend a node on layer $r_1$ and in order to do so we mend a node on layer $r_2$, and so on down to $r_t$. 
	The total distance between the nodes in which we recurse is at most $\sum_{i=1}^{t-1} 2k + r_i - r_{i+1} = (t-1)2k + \sum_{i=1}^{t-1} r_i - r_{i+1} = (t-1)2k + r_i - r_t$. Since $t \leq L$ and $r_i \leq L$, we get that total distance is at most $\le (2k+1)(L-1) = O(\log n)$. 
	
	In order to mend the original node we pay an additional $O(\log n)$ radius in the first phase, and in each recursive step we mend at distance $O(\log n)$ from the nodes in which we recurse. Hence, the total mending radius is $O(\log n)$.
\end{proof}

\newpage

\section{Additional mendability classes in general graphs}\label{app:sqrtn-general}

In this section, we will present an LCL $\Pi$ that has mending radius of $\Theta(\sqrt{n})$. In order to do so, we will present an LCL that is well-defined on any bounded-degree general graph and show that its mending radius is  $O(\sqrt{n})$.  The problem that we will define is such that for every incident edge, each node has an input label assigned to it (note that each edge has two input labels, one for each endpoint), while nodes need to give a single label as output. Then, we prove the existence of a family of graphs, that are grid-like structures, where $\Pi$ requires $\Omega(\sqrt{n})$ to be mended.

In other words, grid-like structures will be the ``hard'' instances for the LCL problem $\Pi$: in fact, we will show that mending $\Pi$ on bounded-degree general graphs that do not have a grid-like structure is not harder than mending $\Pi$ on grid-like instances. In the following we will start by defining what we call a grid-like structure, then we will describe a \emph{local characterization} of such instances, following by the formal description of our LCL problem $\Pi$. Finally, we will show matching upper and lower bounds for $\Pi$. 

\subsection{Grid-like structure}\label{sec:grid-structure}
We will now formally define a set of labels for the nodes, and local constraints over the labels, and we define a family of graphs $\mathcal{G}$ as the one containing all and only the labeled graphs that satisfy the local constraints over all nodes. Informally, a graph $G\in \mathcal{G}$ locally looks like a $2$-dimensional oriented grid (that is, each node has $4$ incident edges consistently labeled with the directions), and globally it wraps around along the two dimensions.

The constant-size set of labels is defined as $\mathcal{L^{\lgrid}} = \{\lup,\ldown,\lleft,\lright$\}.  We denote with $L_e(u)\in \mathcal{L^{\lgrid}}$ the label assigned to the incident edge $e$ of node $u$. We now define some local constraints that, if satisfied over all nodes of a graph, will imply that the graph is a grid-like structure that is properly oriented and wraps around (possibly in some non-aligned way) along the two dimensions.

\paragraph{Local constraints.} 
We define the local constraints $\mathcal{C^{\lgrid}}$ as follows:
\begin{enumerate}
	\item Basic constraints.
	\begin{enumerate}
		\item A node must have degree $4$.
		\item There are no parallel edges or self-loops.
		\item For each label $L_e(u)$, it holds that $L_e(u)\in \mathcal{L^{\lgrid}}$.
		\item Each incident edge has only one label.
		\item For any pair of incident edges $e$, $e'$, for each node $u$, it holds that $L_e(u)\neq L_{e'}(u)$. 
	\end{enumerate}
	\item Structural constraints. Let $z_u(L_1,\ldots,L_k)$ denote the node that we reach starting from node $u$ and following labels $L_1,\ldots,L_k$, where $L_i\in\{ \lup,\ldown,\lleft,\lright \}$ for each $i\in\{1,\ldots,k \}$.
	\begin{enumerate}
		\item If $L_e(u)=\lup$, then $L_e(z_u(\lup))=\ldown$.
		\item If $L_e(u)=\lright$, then $L_e(z_u(\lright))=\lleft$.
		\item $z_u(\ldown,\lright,\lup)=z_u(\lright)$.
		
	\end{enumerate}
\end{enumerate}
In the following, we will refer to a graph $G \in \mathcal{G}$ as a \emph{valid} instance. That is, a labeled graph is a valid instance if the constraints $\mathcal{C^{\lgrid}}$ are satisfied over all nodes.
The family $\mathcal{G}$ of valid instances contains, for example, torus graphs, but it also contains graphs that are not exactly torus graphs (for example grids that wrap around in some unaligned way, differently from a torus). Nevertheless, in some informal sense, the family of valid instances is restricted enough to exclude graphs that are far from being a torus.
Note that, on arbitrary graphs, nodes can check if the constraints $\mathcal{C^{\lgrid}}$ are locally satisfied by exploring a constant-radius neighborhood.

\subsection{The LCL problem \texorpdfstring{\boldmath$\Pi$}{Pi}}

We will now define a constant-radius checkable LCL $\Pi$, for which we will show that it has mending radius $\Theta(\sqrt{n})$, while it can trivially be solved in $0$-rounds of communication. The high-level idea behind this LCL problem $\Pi$ is the following. On the one hand, we allow nodes to output blindly a specific label (which makes the problem trivial to solve), and on the other hand, we design some local constraints on a subset of output labels such that, if we have a partial solution using those labels, then mending requires, in the worst case, $\Theta(\sqrt{n})$. In more detail, we will make sure that, if a node does not output the ``trivial'' label, it must either (1) detect that the graph does not locally look like a valid instance, or (2) it must point towards a node that detects it, forming a proper pointer chain, or (3) it must use pointers to create a cycle that satisfies some special properties. On the upper bound side, we show that we can always mend $\Pi$ in $O(\sqrt{n})$. On the lower bound side, we show that on specific valid instances, that are $\sqrt{n} \times \sqrt{n}$ torus graphs, we can have a partial solution that requires $\Omega(\sqrt{n})$ to mend.

More precisely, we define our LCL $\Pi$ as follows. The set of input labels is defined as $\Sigma =\mathcal{L^{\lgrid}}= \{\lup,\ldown,\lleft,\lright  \}$. The set of output labels is defined as $\Gamma=\{\lzero, \lflag, \lleft, \lup, \lright \}$. Informally, the label $\lzero$ can be used to trivially solve the problem in $0$ rounds, the label $\lflag$ can be used to claim that the graph does not locally look like a valid instance, while labels $\lleft$, $\lup$, and $\lright$ will be used as \emph{pointers}. Let $\loutput(u) \in \Gamma$ denote the output of a node $u$. The local constraints $\mathcal{C^{\llcl}}$ of $\Pi$ are as follows.
\begin{enumerate}
	\item If a node outputs $\lflag$, then there is at least one constraint in $\mathcal{C^{\lgrid}}$ (defined in Section \ref{sec:grid-structure}) that $u$ does not satisfy.
	\item If a node outputs $\lzero$, then there is no neighbor pointing at it. More formally, suppose $\loutput(u)=\lzero$, then, for each neighbor $v$ of $u$, it holds that either
	\begin{itemize}
		\item $\loutput(v)\in \{\lzero,\lflag \}$, or
		\item $\loutput(v)\in \{\lleft, \lup, \lright \}$ such that $z_v(\loutput(v))\neq u$ (i.e., by starting from $v$ and following the edge labeled $\loutput(v)$ we do not reach $u$).
	\end{itemize} 
	\item Two nodes cannot point to each other, and the pointer chains can only end on nodes that output $\lflag$. That is: 
	\begin{itemize}
		\item if $\loutput(u)=\lup$, then $\loutput(z_u(\lup))\in \{\lup, \lleft,\lright,\lflag \}$;
		\item if $\loutput(u)=\lleft$, then $\loutput(z_u(\lleft))\in \{\lleft,\lup, \lflag \}$;
		\item if $\loutput(u)=\lright$, then $\loutput(z_u(\lright))\in \{\lright,\lup, \lflag \}$.
	\end{itemize}
\end{enumerate}
Notice that problem $\Pi$ is a valid LCL problem, since the set of input and output labels have constant size, and all constraints in $\mathcal{C^{\llcl}}$ are locally checkable. Moreover, it is easy to see that $\Pi$ can be trivially solved in $0$ rounds of communication:  the solution  ``all nodes output $\lzero$'' is a valid one.

\paragraph{Upper bound on the mendability of $\Pi$.} Suppose that we need to perform mending at node $u$. We will show that, whatever is the partial solution that we are given, we can mend by modifying the output of nodes up to distance $r = c \cdot \sqrt{n}$, for some large enough constant $c$. We will show this by handling all the scenarios node $u$ may be in, that will boil down into whether or not there is a node $v$ that does not satisfy the constraints in $\mathcal{C^{\lgrid}}$ and can be reached from $u$ by only following edges that are not labeled $\ldown$ for at most $r$ steps. In the first case, we create a pointer chain that starts at $u$ and ends at $v$, while in the second case we create a pointer chain that starts at $u$ and ends at $u$ itself without using the label $\ldown$ (that is, a cycle that passes through $u$ without using $\ldown$). The mending algorithm at node $u$ behaves as follows.

\begin{enumerate}
	
	\item If node $u$ has an inconsistency in its neighborhood, that is, there is at least one constraint in $\mathcal{C^{\lgrid}}$ that is not satisfied by node $u$, then node $u$ is mended with the output label $\lflag$.\label{item:near-inconsistency}
	
	\item Otherwise, if each neighbor of $u$ has a label in $\{\lzero, \lflag, \bot\}$ then node $u$ is mended with the output label $\lzero$.\label{item:freedom}
	
	\item Otherwise, if there exists a neighbor $w$ of $u$ such that 
	\begin{itemize}
		\item $z_u(L)=w$ where $L\in\{\lleft,\lup,\lright\}$, that is, $u$ can reach $w$ by not using label $\ldown$, and
		\item $\loutput(w)= L'\in \{\lleft,\lup,\lright\}$ such that $z_w(L')\neq u$, i.e., $w$ outputs a pointer but it does not point at $u$,
	\end{itemize}
	then we mend $u$ by pointing at $w$, that is $\loutput(u)=L$.\label{item:safe-pointer}
	
	\item Otherwise, if there is a node $v$ within radius $r$ from $u$ such that $v$ has a label in $\{\lflag, \bot\}$ and $v$ can be reached from $u$ by following only edges labeled with $\lleft$, $\lup$, or $\lright$, then we mend node $u$ by creating a pointer chain using output labels in $\{\lleft,\lup,\lright\}$ that starts at $u$ and ends at $v$. More precisely, let $P = (p_1, \ldots, p_k)$ be an arbitrary path where $p_1 = u$ and $p_k = v$, such that $z_{p_i}(L_i) = p_{i+1}$ and $L_i \in \{\lleft,\lup,\lright\}$ for all $i < k$, and all nodes of $P$ are within radius $r$ from $u$. We assign to $p_i$ the label $L_i$ for all $i < k$. \label{item:far-inconsistency}
	
	\item The last case that remains is the one where everything up to radius $r$ from $u$ that can be reached from $u$ using only labels in $\{\lleft,\lup,\lright\}$ looks like a valid instance, that is, all these nodes satisfy the constraints $\mathcal{C^{\lgrid}}$ (and none of these nodes has the label $\bot$). In this case, $u$ is mended by creating a pointer chain using output labels in $\{\lleft,\lup,\lright\}$ that ends up creating a consistently oriented cycle. 
	
	More precisely, let $C = (c_1, \ldots, c_k)$ be a sequence of nodes where $c_1 = c_k = u$, such that $z_{c_i}(L_i) = c_{i+1}$ and $L_i \in \{\lleft,\lup,\lright\}$ for all $i < k$, and all nodes of $C$ are within radius $r$ from $u$. In other words, $C$ is a cycle containing $u$. We assign to $c_i$ the label $L_i$ for all $i < k$. We will later prove that such a sequence of nodes always exists.
	\label{item:no-inconsistency}
\end{enumerate}

\paragraph{Correctness of the mending algorithm.} We now show the correctness of the mending algorithm described above. Let us call a node $v$ that does not satisfy the constraints $\mathcal{C^{\llcl}}$ as an \emph{unhappy} node, and \emph{happy} otherwise. In the following, we will show that, no matter what is the partial solution of $\Pi$ that we are given, we do not create unhappy nodes by mending a node $u$ with the above algorithm.

\begin{itemize}
	\item If we are in case \ref{item:near-inconsistency}, whatever is the label that $u$'s neighbors output, the local constraints $\mathcal{C^{\llcl}}$ of $\Pi$ will continue to be satisfied even after mending $u$ with label $\lflag$. In fact, according to $\mathcal{C^{\llcl}}$, the label $\lflag$ is compatible with all labels in $\Gamma$.
	
	\item If we are in case \ref{item:freedom}, we mend a node $u$ using label $\lzero$ while all neighbors have an output in $\{\lzero, \lflag, \bot\}$, that are all compatible with $\lzero$.
	
	\item If we are in the case described in point \ref{item:safe-pointer}, node $u$ is mended by pointing to a neighbor $w$ who in turn outputs a pointer that does not point to $u$. Clearly, $w$ is still happy after the change. Also, all neighbors $v \neq w$ of $u$ are still happy, since a pointer is always compatible with any label assigned to all nodes it does not point to.
	
	\item Suppose we are in the case described in point \ref{item:far-inconsistency}. All nodes in the path $P$ are happy after the change, since there are no pointers pointing to each other, and the path ends on a node labeled $\lflag$ or $\bot$. Also, all neighbors of nodes in the path are happy, since a pointer is always compatible with any label assigned to all nodes it does not point to.
	
	\item The last case that remains to handle is the one described in point \ref{item:no-inconsistency}. The mending is correct for the same reasoning as for point \ref{item:far-inconsistency}. We need to prove that such a sequence $C$ of nodes always exists. By starting from $u$ and following edges $\lup$ and $\lright$ we must see a grid like structure, since all constraints $\mathcal{C^{\lgrid}}$ are locally satisfied. Note that at some point we must see this grid wrapping around. More formally, there must exist a node $v$ and two paths that start at $u$ and end at $v$ such that the number of followed edges labeled $\lup$ or $\lright$ is different in the two paths. Otherwise, it means that all the visited nodes are different, but this implies that we visit $r^2 = c^2 n > n$ different nodes, which is a contradiction. We can normalize such paths and represent them as a pair $(x,y)$ where we first go $x \ge 0$ times to the right and then $y \ge 0$ times up. Let $w(x,y)$ be the node reached by such a path. We have that there exist two different pairs satisfying $w(x,y) = w(x',y')$. Assume w.l.o.g.\ that $y \ge y'$. Since all constraints $\mathcal{C^{\lgrid}}$ are satisfied, then we also have $w(x,y-y') = w(x',0)$.
	If $x \ge x'$, then we also have that $w(x-x',y-y') = w(0,0) = u$, and this implies that $u$ can reach itself by only following labels $\lup$ and $\lright$. Otherwise, if $x < x'$, we have that $w(0,y-y') = w(x'-x,0)$. We obtain a cycle passing through $u$ by starting from $u$, going up for $y-y'$ steps, and then going left for $x'-x$ steps.
\end{itemize}
The above mending algorithm and its analysis imply the following lemma.

\begin{lemma}\label{lem:upper-bound-general}
	The LCL problem $\Pi$ is $O(\sqrt{n})$-mendable.
\end{lemma}

\paragraph{Lower bound on the mendability.} In the following we will show that the LCL problem $\Pi$ has mending radius $\Omega(\sqrt{n})$. Let $n$ be any square number. The lower bound graph will be a graph $G$ in our family of graphs $\mathcal{G}$. More precisely $G$ will be a $\sqrt{n} \times \sqrt{n}$ $2$-dimensional grid that wraps around in both dimensions in an aligned way, forming a torus graph. More formally, let $(x,y)$ be the coordinates of the node on the $y$-th row and the $x$-th column, where $x,y \in \{ 1,2,\dotsc, \sqrt{n} \}$. There is an edge between any two nodes $u=(x_u,y_u)$ and $v=(x_v,y_v)$ in $G$ if either (1) $x_v=x_u$ and $y_v \equiv y_u + 1 \pmod{\sqrt{n}}$, or (2) $y_v=y_u$ and $x_v \equiv x_u + 1 \pmod{\sqrt{n}}$. Also, it is clear that a torus is a valid instance and it can be labeled with labels in $\mathcal{L^{\lgrid}}$ such that constraints in $\mathcal{C^{\lgrid}}$ are satisfied at all nodes. More precisely, consider an edge $e=\{u,v\}$ and let $u=(x_u,y_u)$ and $v=(x_v,y_v)$. 
The edge $e = \{u,v\}$ is labeled as follows.   
\begin{itemize}
	\item If $x_u=x_v$ and $y_v \equiv y_u + 1 \pmod{\sqrt{n}}$, then $L_e(u) = \lup$ and $L_e(v)=\ldown$.
	\item If $y_u=y_v$ and $x_v \equiv x_u + 1 \pmod{\sqrt{n}}$, then $L_e(u) = \lright$ and $L_e(v)=\lleft$.
\end{itemize}
We construct the following partial solution of the LCL problem $\Pi$ on $G$.

\begin{figure}
	\centering
	\includegraphics[width=0.6\textwidth]{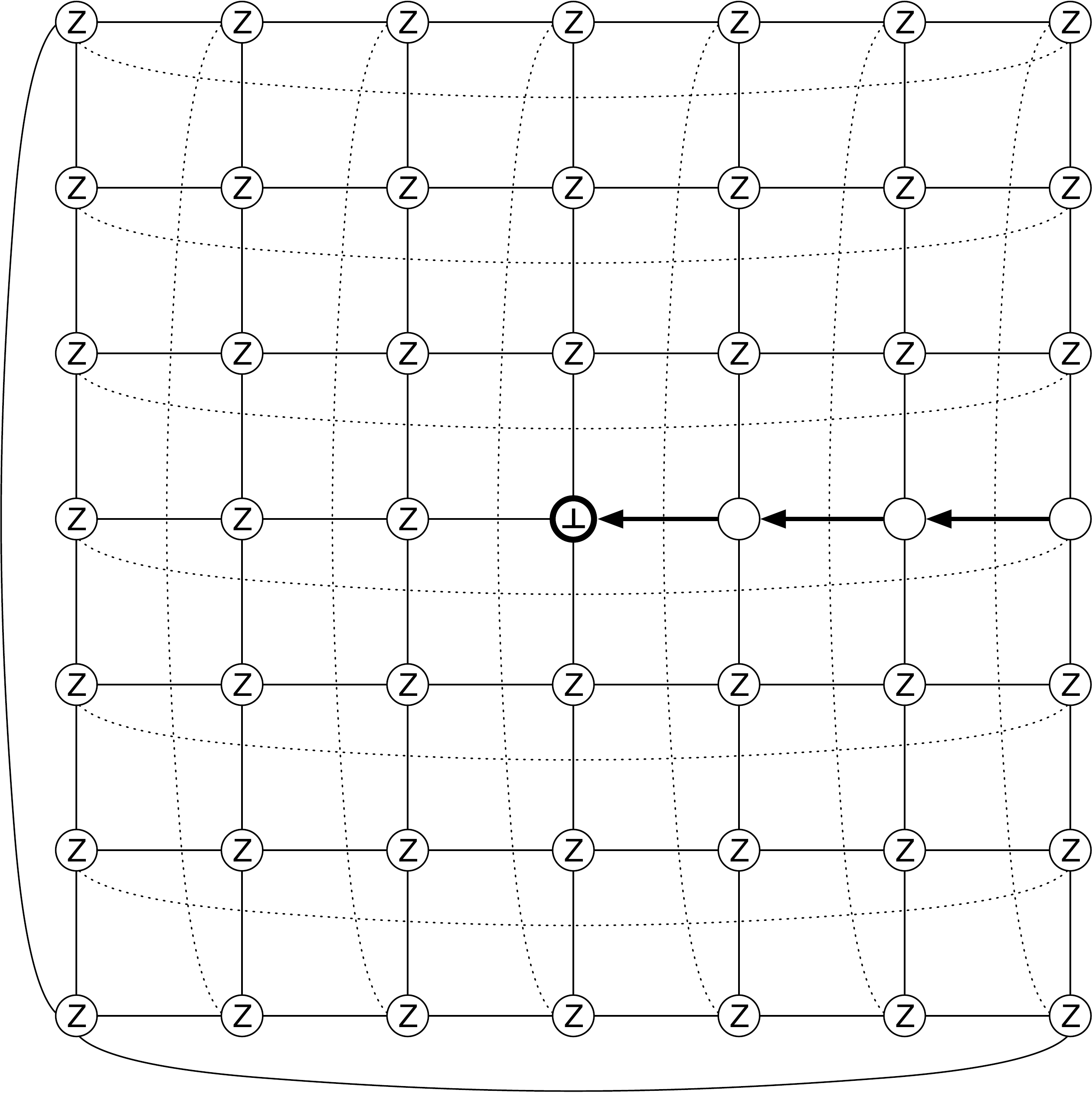}
	\caption{Example of a balanced torus graph that has in input a partial solution of $\Pi$, that is the one that we construct for showing our lower bound on the mending radius; the node that needs mending is the one shown in bold that has label $\bot$, while $\mathsf{Z}$ represents the label $\lzero$.}\label{fig:lower-bound-general-graphs}
\end{figure}

\begin{itemize}
	\item Node $u$ with coordinates $x_u=y_u=\lceil\sqrt{n}/2\rceil$ is labeled $\bot$.
	\item Each node $v$ with coordinates $\lceil\sqrt{n}/2\rceil + 1\le x_v\le \sqrt{n}$ and $y_v= \lceil\sqrt{n}/2\rceil$ is labeled $\lleft$. Let us call $P$ the path containing these nodes.
	\item All other nodes are labeled with $\lzero$.
\end{itemize}
See Figure \ref{fig:lower-bound-general-graphs} for an example. We argue that such a partial solution has a mending radius that is strictly larger than $r=c\sqrt{n}$, for some small enough constant $c$. First of all, note that in such an instance, the only node labeled $\bot$ is $u$, and hence, after mending $u$, all nodes must have a valid label for $\Pi$. Since the graph is a valid instance, the only valid solutions for $\Pi$ are the following:
\begin{enumerate}
	\item all nodes are labeled $\lzero$, or
	\item at least one node has a label in $\{\lleft,\lup,\lright\}$. Since no node can be labeled $\lflag$, pointer chains must form cycles. Also, since no cycle can use label $\ldown$, then any cycle must wrap around on one of the two dimensions, hence the length of each cycle must be at least $\sqrt{n}$.
\end{enumerate}
If the mending procedure tries to obtain a solution of the first type, then it must change the output of all nodes of $P$ to $\lzero$, but this mending has radius larger than $r$.

If the mending procedure tries to obtain a solution of the second type, then it must create a cycle of length at least $\sqrt{n}$. Since the number of nodes labeled $\lleft$ is at most $\sqrt{n}/2$, then the mending procedure has to modify at least $\sqrt{n}/2 > r$ other nodes in order to close a cycle. This implies the following lemma.

\begin{lemma}\label{lem:lower-bound-general}
	The LCL problem $\Pi$ has a mending radius of $\Omega(\sqrt{n})$.
\end{lemma}
Combining the results of Lemma \ref{lem:upper-bound-general} and Lemma \ref{lem:lower-bound-general} we get the following theorem.

\thmreplacesqrtn*

\end{document}